\documentclass[fontsize=11pt,a4paper,numbers=noenddot]{scrartcl}

\usepackage{graphicx}
\usepackage{tabu} 
\usepackage{geometry}
\usepackage{hyperref}
\usepackage{xcolor}
\usepackage{enumerate}
\usepackage{complexity}
\usepackage{blkarray}
\usepackage{wrapfig}
\usepackage{todonotes}
\usepackage{paralist}
\usepackage{listings}
\usepackage{varwidth}
\usepackage{comment}

\usepackage[font=small,labelfont=bf]{caption}
\usepackage[font=small,labelfont=normalfont,labelformat=simple]{subcaption}

\usepackage{listings}
\usepackage{xspace}

\graphicspath{{figs/}}

\usepackage{amsthm}
\usepackage{amsmath}
\usepackage{graphicx}
\usepackage{wrapfig}
\usepackage{hyperref}
\usepackage{enumerate}
\usepackage{complexity}
\usepackage{todonotes}
\usepackage{paralist}
\usepackage{listings}
\usepackage{algorithm}
\usepackage{algpseudocode}
\usepackage{thm-restate}

\def\obstructionFour{\Pi_4^\textnormal{o}}
\def\obstructionFiveA{\Pi_{5,1}^\textnormal{o}}  
\def\obstructionFiveB{\Pi_{5,2}^\textnormal{o}}  
\def\obstructionconvexFiveA{\Pi_{5,1}^\textnormal{oc}}  
\def\obstructionconvexFiveB{\Pi_{5,2}^\textnormal{oc}}  
\def\obstructionhconvexSix{\Pi_{6}^\textnormal{oh}}

\newtheorem{theorem}{Theorem}[section]

\newtheorem{lemma}[theorem]{Lemma}
\newtheorem{corollary}[theorem]{Corollary}
\newtheorem{conj}{Conjecture}
\newtheorem{observation}[theorem]{Observation}
\newtheorem{proposition}[theorem]{Proposition}

\date{}

\title{Using SAT to study plane Hamiltonian substructures in~simple drawings\footnote{%
A special thanks goes to Joachim Orthaber for pointing us to the nice and simple proof of Theorem~\ref{thm:HP_prescribed_edges}. \\
H.~Bergold was supported by DFG-Research Training Group 'Facets of Complexity' (DFG-GRK 2434).
S.~Felsner was supported by DFG Grant FE~340/13-1.
M.~M.~Reddy was supported by Swiss National Science Foundation within the collaborative DACH project \emph{Arrangements and Drawings} as SNSF Project 200021E-171681.
M.~Scheucher was supported by DFG Grant SCHE~2214/1-1.
}}

\def\inst#1{$^{#1}$}

\begin{document}

\author{
Helena Bergold\inst{1}
\and
Stefan Felsner\inst{2}
\and
Meghana M. Reddy\inst{3}
\and
Manfred Scheucher\inst{2}
}

\maketitle

\vspace{-1cm}

\begin{center}
{\footnotesize
\inst{1} 
Department of Computer Science, \\
Freie Universit\"at Berlin, Germany,\\
\texttt{\{helena.bergold\}@fu-berlin.de}
\\\ \\
\inst{2} 
Institut f\"ur Mathematik, \\
Technische Universit\"at Berlin, Germany,\\
\texttt{\{felsner,scheucher\}@math.tu-berlin.de}
\\\ \\
\inst{3} 
Department of Computer Science, \\ 
ETH Z\"urich, Switzerland, \\
\texttt{\{meghana.mreddy\}@inf.ethz.ch}
\\\ \\
}
\end{center}


\begin{abstract}
	In 1988 Rafla conjectured that every simple drawing of a complete
	graph~$K_n$ contains a plane, i.e., non-crossing, Hamiltonian cycle.
	The conjecture is far from being resolved. The lower bounds for plane
	paths and plane matchings have recently been raised to 
	$(\log n)^{1-o(1)}$ and $\Omega(\sqrt{n})$, respectively.  
 
	We develop a SAT framework which allows the study of simple drawings of~$K_n$. Based on
	the computational data we conjecture that every simple
	drawing of~$K_n$ contains a plane Hamiltonian subgraph with $2n-3$
	edges.  We prove this strengthening of Rafla's conjecture 
	for \emph{convex drawings}, a rich
	subclass of simple drawings. Our computer experiments also led to
	other new challenging conjectures regarding plane substructures in
	simple drawings of complete~graphs.

\end{abstract}

\section{Introduction}
\label{sec:intro}

In a \emph{simple drawing}\footnote{In the literature, simple drawings are also called 
	\emph{good drawings}, \emph{simple topological drawings}, and \emph{simple topological graphs}.} of a graph in the plane (resp.\ on the
sphere), the vertices are mapped to distinct points, and edges are drawn as
simple curves which connect the corresponding endpoints but do not contain
other vertices.  Moreover, every pair of edges intersects in at most one
point, which is either a common vertex or a proper crossing (no touching), and
no three edges cross at a common point.  Figure~\ref{fig:simple_obstructions}
shows the obstructions to simple drawings.
In this article, we focus on simple drawings of the complete graph~$K_n$.

\begin{figure}[htb]
	\centering
	\includegraphics{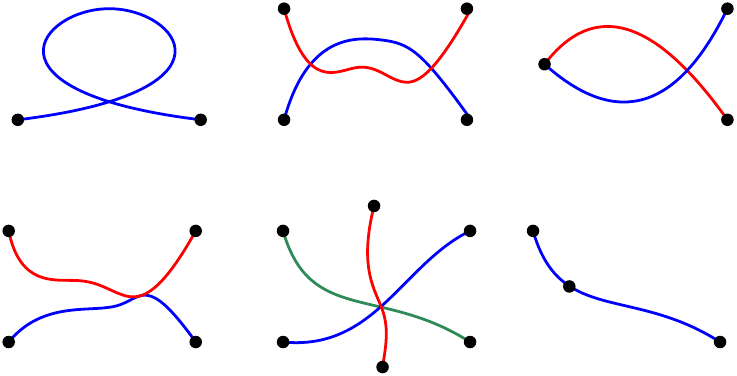}
	\caption{The obstructions to simple drawings.}
	\label{fig:simple_obstructions}
\end{figure}

Problems related to simple drawings have attracted a lot of attention. One of
the reasons is that simple drawings are closely related to interesting classes of
drawings such as crossing-minimal drawings or straight-line drawings, a.k.a.~\emph{geometric drawings}. 
Recently Arroyo et al.~\cite{ArroyoMRS2017_convex} introduced convex drawings
and some subclasses such as $h$-convex drawings. These classes are nested between
simple drawings and geometric drawings.

To study 
substructures in general simple drawings and
subclasses such as convex or $h$-convex drawings, we develop a Python program
which generates a Boolean formula that is satisfiable if and only if there
exists a simple drawing of $K_n$ with
prescribed properties for a specified value of~$n$. Moreover, the solutions of these instances are in
one-to-one correspondence with non-isomorphic simple drawings with the
prescribed properties.  To be more specific, for a given problem of the form
``does there exist a simple drawing of $K_n$ with the property \ldots", the program generates a Boolean formula in
conjunctive normal form (CNF).  We then use state-of-the-art SAT solvers such
as PicoSAT~\cite{Biere08} and CaDiCaL~\cite{Biere2019} to decide
whether a solution exists and to enumerate solutions.  While solutions can
be verified in a straight-forward fashion, 
the correctness in the
\emph{unsatisfiable} case 
has no obvious certificate.  Nevertheless, we verify our unsatisfiability results using the independent proof checking tool DRAT-trim \cite{WetzlerHeuleHunt14}.

In Section~\ref{sec:prelim}, we discuss combinatorial properties of simple
drawings and introduce the convexity hierarchy from Arroyo et al.~\cite{ArroyoMRS2017_convex}. In Section~\ref{sec:all_encoding}, we describe
how we encode simple topological drawings, convexity and other notions in a SAT formula.
Based on computational experiments of small configurations, we came up with
strengthened and modified versions of existing conjectures about simple
drawings. A collection of such conjectures is presented in
Section~\ref{sec:applications}. 
In Section~\ref{sec:applications_planar} we
discuss questions regarding plane substructures in drawings of~$K_n$. 
In a simple drawing~$D$ of~$K_n$,
the subdrawing $D[H]$ induced by a subgraph~$H$ is a \emph{plane substructure} if it does not contain crossing edges.
In particular we focus on variants of Rafla's
conjecture~\cite{Rafla1988}, which asserts that every simple drawing of $K_n$ contains a plane Hamiltonian cycle.  Even though this conjecture
and related substructures attracted the attention of many researchers, so far
only the existence of plane paths of length  $(\log n)^{1-o(1)}$ and plane matchings of
size $\Omega(\sqrt{n})$ is known \cite{SukZeng2022,AichholzerGTVW22twisted}. 	Based on the data for small~$n$, we conjecture that indeed every simple
drawing contains a plane Hamiltonian subdrawing on $2n - 3$ edges.

\begin{restatable}
	{conj}{conjraflatwonminusthree}
	\label{conjecture:rafla_2n_plus_3}
	Every simple drawing of $K_n$ with $n \ge 3$
	contains a plane Hamiltonian subdrawing on $2n-3$ edges.
\end{restatable}

Our main result is a proof of this strengthening of Rafla's conjecture for the class of convex drawings.

\goodbreak

\begin{restatable}
	{theorem}{theoremconvexHC}
	\label{theorem:convex_HC}
	Let $D$ be a convex drawing of $K_n$ with $n\ge 3$ and let $v_\star$ be a
	vertex of~$D$.  Then~$D$ contains a plane Hamiltonian cycle~$C$ which does not
	cross any edge incident to~$v_\star$.  This Hamiltonian cycle can be computed
	in $O(n^2)$ time.  Moreover, if $D$ is $h$-convex, then~$C$ traverses the
	neighbors of $v_\star$ in the order of the rotation system around~$v_\star$.
\end{restatable}

The proof of Theorem~\ref{theorem:convex_HC} is constructive and comes with a
polynomial time algorithm. For the proof we show that convex drawings have a
layering structure and reduce the problem to finding a Hamiltonian path for
every layer.  To simplify the proof and to reduce the number of cases that
have to be considered, we make use of the SAT framework.  We give a sketch of
the proof in Section~\ref{sec:sketch}. The full proof is deferred to
Appendix~\ref{sec:convex_HC_proof}.

While Theorem~\ref{theorem:convex_HC} asserts that, in convex drawings, every
plane spanning star (determined by a vertex) can be extended to a plane
Hamiltonian subdrawing, we have also computational evidence that, in convex
drawings, every plane matching can be extended to a plane Hamiltonian subdrawing
(cf.\ Conjecture~\ref{conjecture:hoffmanntoth_convex}).  An affirmative answer would generalize a result of Hoffmann and
T\'oth~\cite{HoffmannToth2003} about 
geometric drawings.
Another strengthening of Rafla's conjecture was very recently stated by
Aichholzer, Orthaber and Vogtenhuber~\cite{AOV2023}:
for each pair of vertices $a,b$ in a simple drawing of $K_n$ there exists a plane Hamiltonian path 
starting in~$a$ and ending in~$b$. 

We further used the SAT framework to study uncrossed edges
(Section~\ref{sec:uncrossededges})
and 
empty triangles (Section~\ref{sec:emptytriang}), which lead to further 
challenging conjectures. 
In the context of Erd\H{o}s--Szekeres--type problems the  SAT framework also helped in attaining new results, see \cite{BSS2023}.

\section{Preliminaries}
\label{sec:prelim}

Combinatorial properties of simple drawings such as plane substructures do not depend
on the actual drawing of edges but on incidence and order properties of the
drawing.  For a given simple drawing $D$ and a vertex $v$ of $D$, the cyclic
order~$\pi_v$ of incident edges in counterclockwise order around~$v$ is called
the \emph{rotation of $v$} in $D$.  The collection of rotations of all
vertices is called the \emph{rotation system} of~$D$.  In the case of simple
drawings of the complete graph~$K_n$, the rotation of a vertex $v$ is a cyclic
permutation on $V(K_n) \backslash \{v\}$.  The rotation system captures the
combinatorial properties of a simple drawing on the sphere -- the choice of
the outer cell when stereographically projecting the drawing onto a plane has
no effect on the rotation system.

A \emph{pre-rotation system} on $V$ consists of cyclic permutations $\pi_v$ on
the elements $V \backslash \{v\}$ for all $v \in V$.  A
pre-rotation system $\Pi = (\pi_v)_{v \in V}$ is \emph{drawable} if there is a
simple drawing of the complete graph with vertices $V$ such that its rotation
system coincides with~$\Pi$.  Two pre-rotation systems are \emph{isomorphic}
if they are the same up to relabelling and reflection (i.e., all cyclic orders
are reversed). 
Two simple drawings are \emph{weakly isomorphic}
if their rotation systems are isomorphic.  
On four vertices there are three non-isomorphic pre-rotation systems. 
The $K_4$ has exactly
two non-isomorphic simple drawings on the sphere: the drawing with no crossing
and the drawing with one crossing.  Hence, the two corresponding pre-rotation systems are drawable, and the third pre-rotation system is an obstruction to drawability.  
It is denoted by $\obstructionFour$ and shown 
in Figure~\ref{fig:rotsys_obstructions}.

For a pre-rotation system $\Pi = (\pi_v)_{v \in V}$ and a subset of the
elements $I \subset V$, the \emph{sub-configuration induced by $I$} is
$\Pi|_I = (\pi_v|_I)_{v \in I}$, where $\pi_v|_I$ denotes the cyclic permutation
obtained by restricting $\pi_v$ to $I \backslash \{v\}$.  A pre-rotation
system $\Pi$ on $V$ \emph{contains} $\Pi'$ if there is an induced
sub-configuration $\Pi|_I$ with $I \subseteq V$ isomorphic to~$\Pi'$.
A pre-rotation system not containing $\Pi'$ is called
\emph{$\Pi'$-free}.

\begin{figure}[htb]
	\centering
	\includegraphics{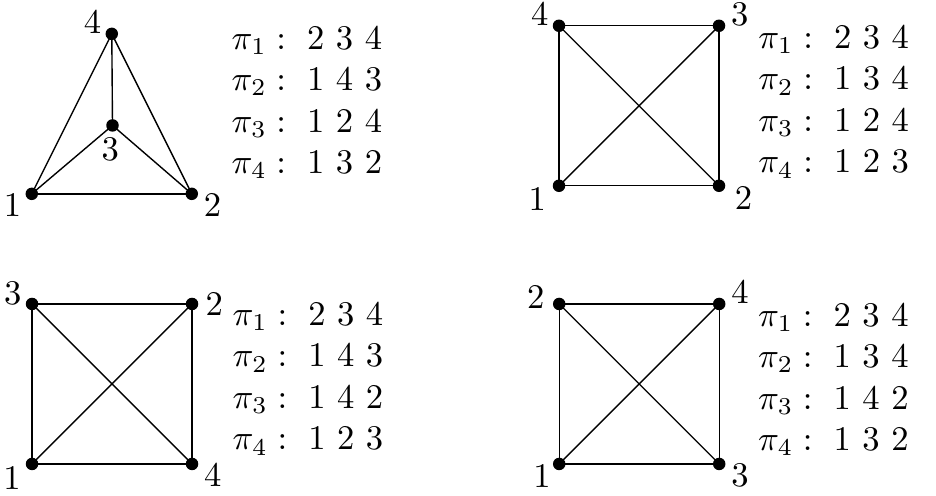}
	\caption{The four rotation systems on $4$ elements with their drawings. The first one is the only non-isomorphic one.}
	\label{fig:rs_n4_noniso}
\end{figure}

A crossing pair of edges involves four vertices. 
By studying drawings
of the $K_4$ (see Figure~\ref{fig:rs_n4_noniso}) we learn that a crossing pair of edges
can be identified from the underlying rotation system. 
Hence, the pairs of crossing edges in a drawing of $K_n$ are fully determined by its underlying rotation system.
We summarize: 

\goodbreak

\begin{observation}			
	\label{observation:basics}
	The following two statements hold:
	\begin{compactenum}[(i)]
		\item 
		\label{item:obstructionfour_notdrawable}
		A pre-rotation system containing $\obstructionFour$ is not drawable.    
		\item 
		\label{item:obstructionfour_crossingsdetermined}
		Let $\Pi$ be a $\obstructionFour$-free pre-rotation system on $[n]$. The
		subconfiguration induced by a 4-element subset is drawable and determines which pairs of
		edges cross in the drawing.
	\end{compactenum}
\end{observation}

Note that part (\ref{item:obstructionfour_crossingsdetermined}) of the lemma allows to talk about the crossing pairs of edges of a $\obstructionFour$-free pre-rotation system, even if there is no associated drawing.

\'Abrego et al.~\cite{AbregoAFHOORSV2015} generated all pre-rotation systems
for up to $9$ vertices and used a drawing program based on back-tracking to
classify the drawable ones. In particular, they provided the following
classification:

\begin{restatable}[\cite{AbregoAFHOORSV2015}]{proposition}{propclassificationRS}
	\label{proposition:rotsys_classification_n6}
	A pre-rotation system on $n \le 6$ elements
	is drawable if and only if it does not contain
	$\obstructionFour$, $\obstructionFiveA$, or $\obstructionFiveB$
	(Figure~\ref{fig:rotsys_obstructions}) as a subconfiguration.
\end{restatable}

\begin{figure}[htb]
	\centering
	\includegraphics{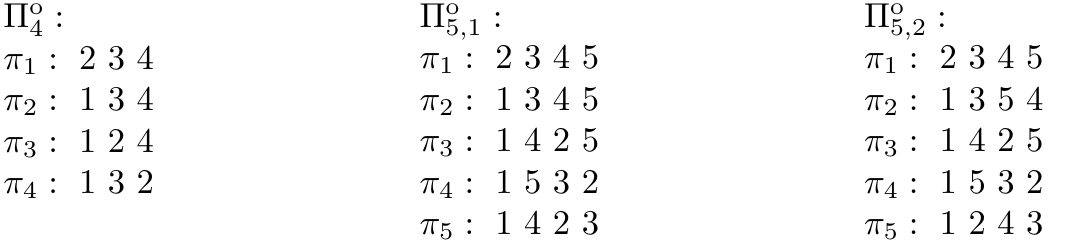}
	\caption{The three obstructions $\obstructionFour$, $\obstructionFiveA$, and $\obstructionFiveB$ for rotation systems.}
	\label{fig:rotsys_obstructions}
\end{figure}

Kyn\v{c}l showed that a pre-rotation system is drawable if and only if all induced 4-, 5-, and 6-element subconfigurations are drawable \cite[Theorem 1.1]{Kyncl2020}.
Together with \autoref{proposition:rotsys_classification_n6} this 
yields the following characterization:

\begin{restatable}
	{theorem}{thmclassificationRS}
	\label{theorem:rotsys_5tuples_characterization}
	A pre-rotation system on $n$ elements is drawable if and only if
	it does not contain $\obstructionFour$, $\obstructionFiveA$ or $\obstructionFiveB$ (Figure~\ref{fig:rotsys_obstructions})
	as a subconfiguration.
\end{restatable}

Details can be found in Appendix~\ref{sec:characterizationRS}.
There in particular we present an
alternative proof of Proposition~\ref{proposition:rotsys_classification_n6}
for which we use a two-level SAT approach: First we use SAT solver to  enumerate pre-rotation systems. Then for each pre-rotation system
we test drawability by using an auxiliary SAT encoding for planar graphs which is based on Schnyder's
characterization \cite{Schnyder1989}.
For more information and further literature on planarity encodings 
we refer to \cite{ChimaniHW19} and~\cite{KirchwegerSS2023}.

We now turn to the definition of convex drawings. The notion of convexity is
based on the \emph{triangles} of a drawing, i.e., the subdrawings induced by
three vertices.  Since in simple drawings the edges of a triangle do not
cross, a triangle partitions the plane (resp.\ sphere) into exactly two
connected components.  The closures of these components are the two
\emph{sides} of the triangle.  A side $S$ of a triangle is \emph{convex} if
every edge that has its two vertices in~$S$ is fully drawn inside~$S$.  A
simple drawing of the $K_n$ is \emph{convex} if every triangle has a convex
side.  Moreover, a convex drawing is \emph{$h$-convex} (short for hereditary
convex) if the side~$S_1$ of a triangle which is fully contained in the convex
side~$S_2$ of another triangle is convex.  
For further aspects of the
convexity hierarchy we refer to
\cite{ArroyoMRS2017_convex,ArroyoMRS2017_pseudolines,BFSSS_TDCTCG_2022}.

Arroyo et al.~\cite{ArroyoMRS2017_convex} showed that convex and $h$-convex
drawings can be characterized via finitely many forbidden subconfigurations.
A simple drawing is \emph{convex} if and only if it does not contain
$\obstructionconvexFiveA$ or $\obstructionconvexFiveB$ (cf.
Figure~\ref{fig:rotsys_obstructions_convex}) as a subconfiguration.
Moreover, a convex drawing is \emph{$h$-convex} if and only if it does not
contain $\obstructionhconvexSix$ (cf.
Figure~\ref{fig:rotsys_obstructions_hconvex}) as a subconfiguration.

\begin{figure}[htb]
	\centering
	\includegraphics[page=2]
	{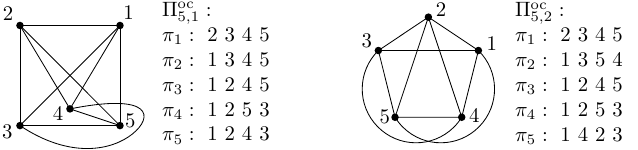}
	\caption{The two obstructions $\obstructionconvexFiveA$ (left) and $\obstructionconvexFiveB$ (right) for convex drawings.}
	\label{fig:rotsys_obstructions_convex}
\end{figure}

\begin{figure}[htb]
	\centering
	\includegraphics[page=3]
	{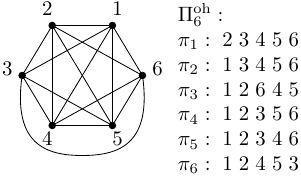}
	\caption{The obstruction $\obstructionhconvexSix$ for $h$-convex drawings.}
	\label{fig:rotsys_obstructions_hconvex}
\end{figure}

\section{SAT encoding for rotation systems}
\label{sec:all_encoding}

To encode pre-rotation systems on $[n]$
as a Boolean satisfiability problem for a constant~$n$,
we use Boolean
variables and clauses to encode the rotations of the vertices. 
More specifically, we introduce a Boolean variable $X_{aib}$ for every pair of distinct vertices $a, b \in [n]$ and every index $i \in [n-1]$, to indicate whether $\pi_a(i)=b$,
and use appropriate clauses to assert that for every $i$ the variables $X_{aib}$ indeed encode a cyclic permutation.
To restrict the search space to
rotation systems of the complete graph $K_n$, that is,
pre-rotation system that are drawable, 
we introduce clauses to forbid the drawability-obstructions given in Theorem~\ref{theorem:rotsys_5tuples_characterization}. 
With this encoding, every
solution of the SAT formula corresponds to a rotation system 
of a simple	drawing of~$K_n$.  
Moreover, we encode convexity, edge crossings, plane subdrawings and
further notations in terms of the rotation systems and implement a static symmetry breaking that
ensures that the solutions are in one-to-one correspondence with isomorphism classes of rotation systems.
Details are deferred to Appendix~\ref{sec:encoding}.

\section{Substructures of rotation systems}
\label{sec:applications}

Here, we discuss and present new results and new conjectures based on the computational data of the SAT framework.

\subsection{Plane Hamiltonian substructures}
\label{sec:applications_planar}

Many questions about simple drawings concern plane substructures, i.e., they ask about the existence of
a crossing-free subdrawing with certain properties. 
One of the most prominent conjectures in this direction is by Rafla.

\begin{conj}[{\cite{Rafla1988}}]
	\label{conjecture:rafla}
	Every simple drawing of $K_n$ with $n\ge 3$ contains a plane Hamiltonian cycle.
\end{conj}

Rafla verified the conjecture for $n \le 7$.  Later \'Abrego
et al.\ \cite{AbregoAFHOORSV2015} enumerated all rotation
systems for \mbox{$n \le 9$} and verified the conjecture for \mbox{$n \le 9$}.  Furthermore
the conjecture was proven for some particular subclasses such as geometric,
monotone, cylindrical drawings and c-monotone\cite{AOV2023}. 
With the SAT framework we verify Conjecture~\ref{conjecture:rafla} for all $n \leq 10$. Technical
details are given in Appendix~\ref{plane_substructures}. 

We continue with
variants and a strengthening of the conjecture. 
Recently, Suk and Zeng \cite{SukZeng2022} and Aichholzer et al.\
\cite{AichholzerGTVW22twisted} independently showed that simple drawings
contain a plane path of length $(\log n)^{1-o(1)}$.  Suk and Zeng
showed that every simple drawing of $K_n$ contains a plane copy of every tree
on $(\log n)^{1/4-o(1)}$ vertices.  Aichholzer et al.\ moreover showed the
existence of a plane matching of size $\Omega(n^{1/2})$,  improving previous
bounds from 
\cite{PachSolymosiToth2003,PachToth2005,FoxSudakov2009,Suk2012,FulekRuizVargas2013,Fulek2014,RuizVargas17},
and 
that generalized twisted drawings contain a plane Hamiltonian cycle if $n$ is odd.

Fulek and Ruiz--Vargas \cite[Lemma~2.1]{FulekRuizVargas2013} showed that every
simple drawing of $K_n$ contains a plane subdrawing with $2n-3$ edges, which
is best-possible as witnessed by the geometric drawing of
$n$ points in convex position (see Figure~\ref{fig:convex_C5_and_twisted_T5}(left)).  
While for a prescribed connected plane spanning subdrawing, an augmentation with the maximum number of edges can be computed in cubic time, 
it is NP-complete to determine the size of the largest plane subdrawing~\cite{GarciaTejelPilz2021}.

\begin{figure}[htb]
	\centering
	\includegraphics{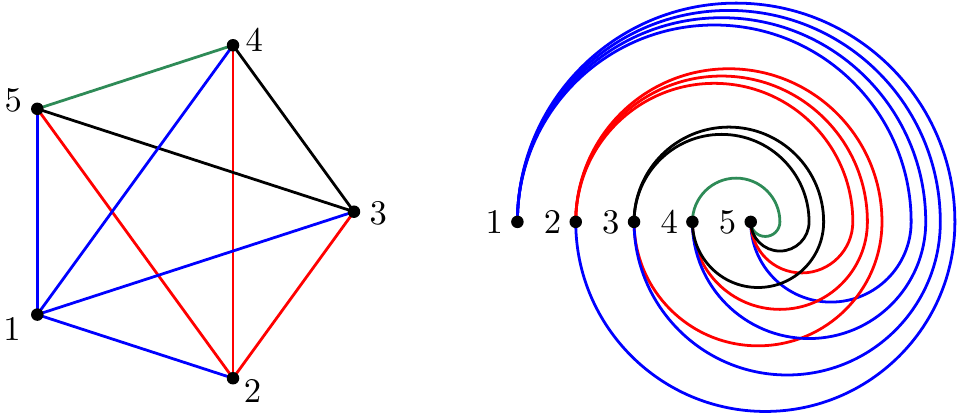}
	\caption{
		The two crossing-maximal simple drawings of~$K_5$: 
		(left)~a geometric drawing of the \emph{perfect convex}~$C_5$ 
		and 
		(right)~the \emph{perfect twisted} drawing~$T_5$ isomorphic to $\obstructionconvexFiveA$.
	}
	\label{fig:convex_C5_and_twisted_T5}
\end{figure}

We investigated a combination of both properties, i.e., for plane Hamiltonian
subdrawings of size $2n-3$.  Such a subdrawing exists for all simple drawings
with $n \leq 8$ (see Appendix~\ref{plane_substructures}) and we conjecture
that this pattern continues.

\conjraflatwonminusthree*

For convex drawings, we proved a strengthened version of Conjecture~\ref{conjecture:rafla_2n_plus_3}, where the plane Hamiltonian subdrawing contains a spanning star. 
See Section~\ref{sec:sketch} for a sketch and Appendix~\ref{sec:convex_HC_proof} for the full proof.

\theoremconvexHC*

The Hamiltonian cycle from Theorem~\ref{theorem:convex_HC} together with the
spanning star centered at~$v_\star$ form a plane subdrawing on $2n-3$ edges.
The statement of Theorem~\ref{theorem:convex_HC} is not true for
non-convex drawings. For example, if $v_\star$ is chosen as vertex $5$ of $T_5$ Figure~\ref{fig:rotsys_obstructions_convex}(left),
then the only edges not crossing star-edges are $\{1,2\}$ and  $\{2,3\}$. 
Moreover, Figure~\ref{fig:rotsys_obstructions_hconvex} shows a convex
drawing which is not $h$-convex and where no plane Hamiltonian cycle traverses
the neighbors of $v_\star=3$ in the order of the rotation system
around~$v_\star$.

\paragraph{Extending Hamiltonian cycles}
Another way to read Theorem~\ref{theorem:convex_HC} is the following: given a spanning star in a convex drawing, we can extend this star to a plane Hamiltonian subdrawing on $2n-3$ edges.
As a variant of this formulation, 
we tested whether the other direction is true, i.e., 
whether any given plane Hamiltonian cycle can be extended to a plane subdrawing with $2n-3$ edges. 
It is worth noting that extending the Hamiltonian cycle by a star already fails in the geometric setting; see e.g.\  Figure~\ref{fig:counterexHCextensionbystar}. 
While for convex drawings such an extension exists for all $n \leq 10$ (see Appendix~\ref{plane_substructures} for technical detail), for general simple drawings it does not; see the example on 8 vertices depicted in Figure~\ref{fig:HTP_counter}(left).

\begin{conj}
	\label{conjecture:extend_HC_to_2n_3}
	Let $D$ be a convex drawing $D$.
	Then every plane Hamiltonian cycle~$C$
	can be extended to a plane Hamiltonian subdrawing on $2n-3$ edges. 
\end{conj}

\begin{figure}[htb]
	\centering
	\includegraphics{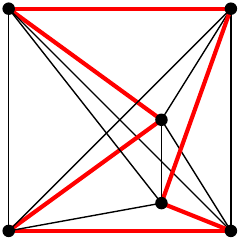}
	\caption{
		A plane Hamiltonian cycle (red) which cannot be extended by a spanning star.
	}
	\label{fig:counterexHCextensionbystar}
\end{figure}

\paragraph{Hamiltonian cycles avoiding a matching}
Another variant is to consider a 
prescribed matching and
ask whether there is a Hamiltonian cycle which together with the
matching builds a plane Hamiltonian substructure, i.e., the edges of the matching are
not crossed by the Hamiltonian cycle (but possibly contained).  Hoffmann and T\'oth
\cite{HoffmannToth2003} showed that
for every plane perfect matching~$M$ in a geometric drawing of $K_n$ there
exists a plane Hamiltonian cycle that 
does not cross any edge from~$M$.  While the statement does not generalize to
simple drawings (see Figure~\ref{fig:HTP_counter}(right)), it seems to
generalize to convex drawings.  Additionally, we consider arbitrary plane matchings and dare the following conjecture, which we verified for $n \leq 11$; 
see Appendix~\ref{plane_substructures}.

\begin{conj}
	\label{conjecture:hoffmanntoth_convex}
	For every plane matching~$M$
	in a convex drawing of~$K_n$ 
	there exists a plane Hamiltonian cycle 
	that does not cross any edge from~$M$.
\end{conj}

\begin{figure}[htb]
	\centering
	
	\centering
	\includegraphics[scale =0.8]
	{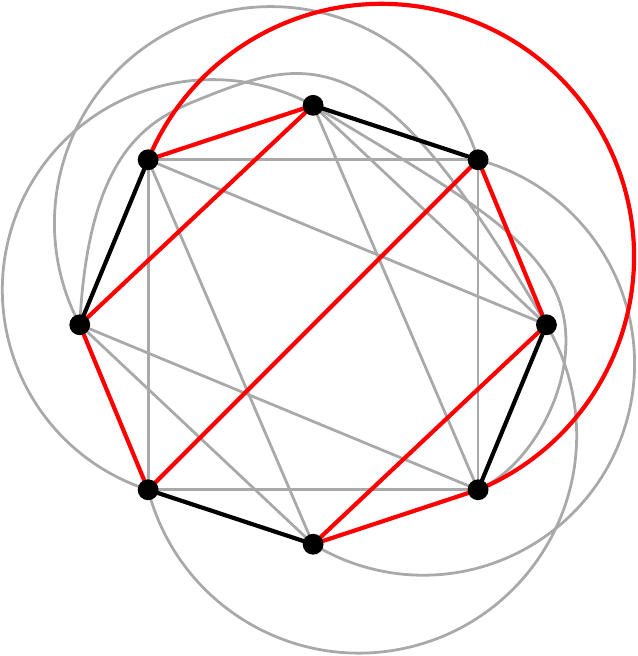}
	\quad
	\quad
	\centering
	\includegraphics[scale = 0.7]{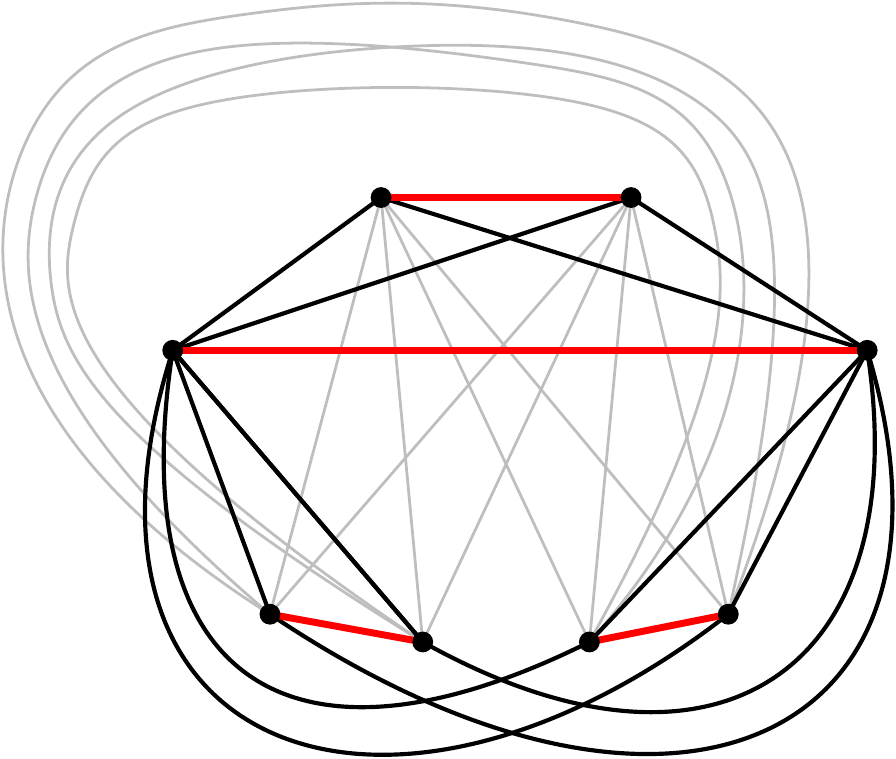}
	
	\caption{
		(left)~A plane Hamiltonian cycle of a drawing of $K_8$ (red) which can only be extended by 4 edges (black);
		(right)~A perfect matching (red)
		in a non-convex drawing of $K_8$, 
		which does not contain a plane Hamiltonian cycle not crossing the matching edges. Edges which cross matching edges are marked in grey.
	}
	\label{fig:HTP_counter}
	
\end{figure}

\paragraph{Hamiltonian paths with a prescribed edge}
In general it is not possible to find a Hamiltonian cycle containing a
prescribed edge. For example, the edge $\{1,3\}$ of $C_5$
depicted in Figure~\ref{fig:convex_C5_and_twisted_T5} is not contained in any
Hamiltonian cycle. Hence, prescribing an edge of a Hamiltonian cycle is not
even possible in the geometric setting. 
Instead we only ask
for a Hamiltonian path containing a prescribed edge. Again there is an easy example:
consider the edge $\{1,5\}$ in $T_5$ depicted in Figure~\ref{fig:convex_C5_and_twisted_T5}. It is not
contained in any Hamiltonian path. 
However, prescribing an edge of a Hamiltonian path is possible in convex drawings.
The proof follows from Theorem 1 and is deferred to Appendix~\ref{app:HP_prescribed_edges}.

\begin{theorem}
	\label{thm:HP_prescribed_edges}
	Let $D$ be a convex drawing of $K_n$ and let $e$ be an edge of~$D$.  Then $D$
	has a plane Hamiltonian path containing the edge~$e$.
\end{theorem}

\subsection{Uncrossed edges}
\label{sec:uncrossededges}
In this section we focus on edges which
are not crossed by any other edges of a given drawing of~$K_n$.  In 1964,
Ringel~\cite{Ringel1964} proved that in every simple drawing of $K_n$ there
are at most $2n-2$ edges without a crossing.  Later, Harborth and Mengersen
\cite{HarborthMengersen1974} studied the minimal number of uncrossed edges and
showed that every simple drawing of $K_n$ with $n\le 7$ contains an uncrossed
edge and presented drawings for $n \geq 8$ where every edge is crossed.
However, these drawings are not convex.
Using the SAT framework (see Appendix~\ref{uncrossed}), we 
reproduced these
results.  Moreover, when restricted to convex drawings, all
drawings have an uncrossed edge for $n \le 10$. 
For $n = 11, \ldots, 21$, there are convex
drawings where every edge is crossed.  These examples are indeed
\mbox{$h$-convex}.  Thus, we conjecture:

\begin{conj}
	\label{conjecture:hconvex_aec}
	For every $n \ge 11$ there exists an $h$-convex drawing on
	$K_n$ where every edge is crossed.
\end{conj}

\subsection{Empty triangles}
\label{sec:emptytriang}

An \emph{empty triangle} in a simple drawing of the $K_n$ is a triangle
induced by three vertices such that the interior of one of its two sides does
not contain any vertex.  
Let $\triangle(n)$ denote
the minimum number of empty triangles in a simple drawing of $K_n$. Harborth \cite{Harborth98} proved
the bounds $2 \le \triangle(n) \le 2n-4$ and asked whether
$\triangle(n) = 2n-4$.  Ruiz--Vargas \cite{RuizVargas2013} proved  $\triangle(n) \ge \frac{2n}{3}$. 
The bound was further improved by Aichholzer et al.\
\cite{AichholzerHPRSV2015} to $\triangle(n) \ge n$.
Moreover, \'Abrego et al.\ \cite{AichholzerHPRSV2015,AbregoAFHOORSV2015} used
enumeration to show that $\triangle(n) = 2n-4$ for all
$n \leq 9$.  This result was reproduced with the SAT framework.

Aichholzer et al.\ \cite{AichholzerHPRSV2015} defined \emph{lucky} vertices, show that drawings with a lucky vertex have $2n-4$ empty triangles
and present a drawing $\Pi_{8}^\text{unlucky}$ of $K_8$ without lucky vertices, see \cite[Figure~7]{AichholzerHPRSV2015}.
With the SAT framework, we show that $\Pi_{8}^\text{unlucky}$ is the only
unlucky example on at most $9$ vertices.  For $n=10$ there is another
unlucky example, which contains $\Pi_{8}^\text{unlucky}$. 
We dare the following conjecture.

\begin{conj}
	\label{conjecture:unlucky}
	Every unlucky drawing contains $\Pi_{8}^\textnormal{unlucky}$.
\end{conj}

Arroyo et al.\
\cite{ArroyoMRS2017_pseudolines} showed that every convex drawing of $K_n$ contains at least
$\triangle_{\text{conv}}(n) \ge \Omega(n^2)$ empty triangles, which is asymptotically tight as
there exist geometric drawings with $(2+o(1))n^2$ empty
triangles~\cite{BaranyFueredi1987}. Determining the value of $\triangle_{\text{geom}}(n)$, i.e., the 
minimum number of empty triangles in geometric drawings of~$K_n$, is a 
challenging problem cf.~\cite[Chapter~8.4]{BrassMoserPach2005};
see \cite{ABHKPSVV2020_JCTA} for the currently best bound. 
In a recent paper, García et al.\
\cite{GarciaTejelVogtenhuberWeinberger2022} showed that generalized twisted drawings have exactly
$2n-4$ empty triangles.

\section{Sketch of the proof of Theorem~\ref{theorem:convex_HC}} 
\label{sec:sketch}

In this section we give a sketch of the proof of
Theorem~\ref{theorem:convex_HC}.  The full proof is deferred to
Appendix~\ref{sec:convex_HC_proof}.  We prove the existence of the Hamiltonian
cycle in a constructive way and present a quadratic time algorithm which, for
a given convex drawing $D$ of the complete graph $K_n$ and a vertex
$v_{\star}$, computes a plane Hamiltonian cycle which does not cross edges
incident to $v_{\star}$.  Since this statement is independent from the
labeling of the vertices, we assume throughout the section that
$v_{\star} = n$ and that the other vertices are labeled from $1$ to $n-1$ in
cyclic order around~$v_{\star}$, i.e., according to the rotation $\pi_{n}$.
For the labels of vertices $1$ to $n-1$ we use the arithmetics modulo $n-1$. 
Since the convexity of a drawing in the plane is independent of the choice
of the outer face, we choose a drawing where $v_{\star}$ belongs to the outer face. 

Denote the vertex $v_{\star}=n$ as \emph{star vertex} and edges incident to~$v_{\star}$ as
\emph{star edges}.  Note that every Hamiltonian cycle contains exactly two
star edges.  We let $\hat{E}$ be the set of non-star edges and call an
edge $e \in \hat{E}$ \emph{star-crossing} if it crosses a star edge.
For the Hamiltonian cycle we only use edges which are not star-crossing.
Let $E^\circ=\hat{E} \setminus \{e \in \hat{E} : e \text{ is star-crossing}\}$
be the set of these candidate edges.
The following lemma follows from the properties of a simple drawing.

\begin{restatable}{observation}{lemmaHCisPlane}
	\label{lem:HCplane}
	Let $e = \{u,v\}$ and $e' = \{u',v'\}$ be independent edges from
	$E^\circ$. If 
	$u,v,u',v'$ appear in this cyclic order
	around~$n$, then $e$ and $e'$ do not cross.
\end{restatable}

Particular focus will be on the edges $e = \{v, v+1\}$ with $1\leq v < n$. 
Such an edge is called \emph{good} if it is
not star-crossing. Otherwise, if
edge $b = \{v, v+1\}$ crosses a star edge $\{w,v_{\star}\}$, then we say that $b$
is a \emph{bad edge} with \emph{witness} $w$.

If there is at most one bad
edge $\{v,v+1\}$, then the $n-2$ good edges together with the
two star edges $\{v,v_{\star}\}$ and $\{v+1,v_{\star}\}$ form a
Hamiltonian cycle which visits the non-star in the order of the rotation around the star vertex.

For the proof of the theorem it is critical to understand the structure of bad edges. 
A drawing with more than one bad edge is shown in
Figure~\ref{fig:rotsys_obstructions_hconvex}, where
the choice of $v_\star=3$ yields the two bad edges $\{2,6\}$ and $\{4,6\}$.
We start with a lemma concerning the convex side of the triangle
spanned by~$v_{\star}$ and a bad edge.
If $b = \{v, v+1\}$ is a bad edge with witness $w$, then  the side
of the triangle spanned by $\{v, v+1, v_{\star}\}$ that contains $w$ is not
convex, since $v_{\star}$ and $w$ both belong to this side but the edge
$\{v_{\star}, w\}$ is not fully contained in this side.

\begin{restatable}{observation}{obsConvexSide}
	\label{obs:ConvexSide}
	Let $b = \{v, v+1\}$ be a bad edge, then the side
	of the triangle spanned by $\{v, v+1, v_{\star}\}$ containing the witnesses is not convex. 
\end{restatable}

Every triangle in a convex drawing has at least one convex side. Thus, the triangle spanned by $\{v, v+1, v_{\star}\}$ has exactly one convex side which is the side not containing the witnesses.

The following lemma will be the key to understand the structure of bad edges.
To avoid a large case distinction, we used SAT to show that there is no convex
drawing on at most 7 vertices which violates the statement. Since convex
drawings have a hereditary structure 
(i.e., the drawing stays convex when a vertex and its incident edges are removed), 
this proves the lemma.

\begin{restatable}[Computer]{lemma}{lemmaBadEdgesNested}
	\label{lem:twobadedges1}
	Let $b = \{v, v+1\}$ and $b' = \{v', v'+1\}$ be two distinct bad edges with
	witnesses $w$ and $w'$, respectively.  
	Then $w \neq w'$. 
	Further $w',w,v,v'$ appear in this or the reversed cyclic
	order around~$v_{\star}$.  
	Moreover, each of the triangles $\{v,v+1,v_{\star}\}$ and $\{v',v'+1,v_{\star}\}$ is in the
	convex side of the other.
\end{restatable}

This lemma directly implies that there is at most one bad edge
in an $h$-convex drawing, which implies the moreover part of \autoref{theorem:convex_HC}.

For the case that a drawing has two or more bad edges the lemma implies that there is a partition of the vertices $1, \ldots, n-1$ into two blocks
such that in the rotation around $v_{\star}$ both blocks of the partition are consecutive, one contains the vertices of all bad edges 
and the other contains all the witnesses. 
Let $b=\{v,v+1\}$ be the bad edge whose vertices are last in the
clockwise order of its block. 
We cyclically relabel the vertices such that $b$ becomes $\{n-2,n-1\}$. This makes the labels of all witnesses smaller than
the labels of vertices of bad edges. In particular we then have the following two properties:

\begin{itemize}
	\item[\emph{(sidedness)}] If $\{v,v+1\}$ is a bad edge with
	witness $w$, then $w < v$.
	\item[\emph{(nestedness)}] If $b = \{v, v+1\}$ and $b' = \{v', v'+1\}$ are
	bad edges with respective witnesses $w$ and $w'$ and if $v < v'$, then
	$w' < w$.
\end{itemize}

In addition we can apply a change of the outer face (via stereographic
projections) so that the vertex $v_{\star}$ and the initial segments of the edges $\{v_\star,1\}$ and
$\{v_\star,n-1\}$ belong to the outer face.  

The nesting property implies that we can label the bad edges 
as $b_1, \ldots, b_m$ for some $m \ge 2$, such that if $b_i = \{v_i,v_i +1\}$, then
$1 < v_1 < v_2 < \ldots < v_m = n-2$.
Moreover, we denote by $w_i^L$ and $w_i^R$ the leftmost (smallest) and the rightmost (largest) witness of the bad edge $b_i$, respectively.
Then 
$ 1 \leq w_{m}^L \leq w_{m}^R  < w_{m-1}^L \leq w_{m-1}^R < \ldots
< w_{1}^L \leq w_{1}^R  $.
Sidedness additionally implies $w_{i}^R < v_i$.
Figure~\ref{fig:layeringbadedgeswithnotation} shows the situation for two bad edges
$b_i$ and $b_{i+1}$.
Note that $v_i+1 = v_{i+1}$ is possible. 

\begin{figure}[htb]
	\centering
	\includegraphics[scale =0.9]{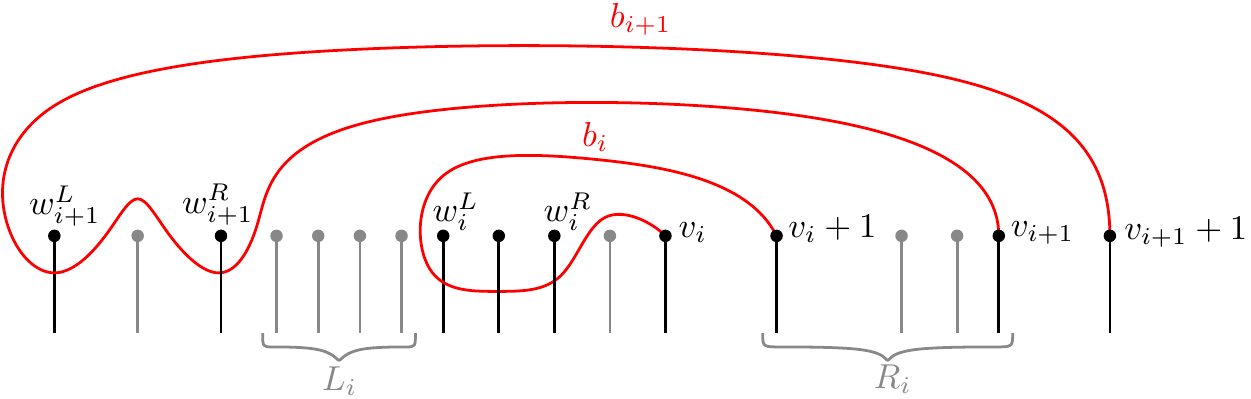}
	\caption{An illustration of sidedness and nesting for two bad edges.}
	\label{fig:layeringbadedgeswithnotation}
\end{figure}

Let $L_i = \{x \in [n-1]: w_{i+1}^R < x < w_i^L\}$ and
$R_i = \{x \in [n-1]: v_{i} +1\leq x \leq v_{i+1}\}$ denote the left and
the right blocks of vertices between two consecutive bad edges $b_i$ and
$b_{i+1}$, see Figure~\ref{fig:layeringbadedgeswithnotation}.

With the following lemma whose proof is deferred to the appendix
we identify useful edges in $E^\circ$.
An illustration of the proof is given in Figure~\ref{fig:rightmostwitness_sketch}.
\begin{restatable}{lemma}{lemTwoUsefullEdges}
	\label{lem:leftrightmostwitness}
	For all $i \in [m]$, neither of the edges $\{w_i^L, v_i+1\}$ and $\{w_i^R, v_i\}$ is star-crossing.
\end{restatable}

\begin{figure}[htb]
	\centering
	
		\hbox{}\hfill
		\includegraphics[page =1, scale =0.7]{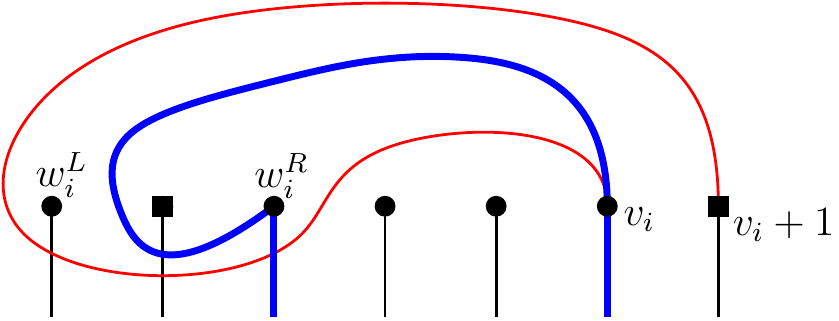} 
		\hfill 
		\includegraphics[page =2, scale = 0.7]{figs/insidebadedge.pdf}
		\hfill\hbox{}
	
	\caption{
		If the blue edges $\{w_i^R,v_i\}$, $\{w_i^L,v_i+1\}$ are star-crossing, the drawing cannot be convex.
		The vertices marked with squares are witnesses for the corresponding side of the blue triangle not being convex.
	}
	\label{fig:rightmostwitness_sketch}
\end{figure}

The general idea to  construct a plane Hamiltonian cycle
is as follows: Begin with the edge $\{v_\star,v_1\}$. When vertex $v_i$ is visited,
we go from $v_i$ to $w_{i}^R$ and then one by one with decreasing labels to $w_i^L$.
From there we want to visit
$v_i+1$ and then one by one in increasing order through the vertices of $R_i$ until we reach $v_{i+1}$. When we reach $v_m+1=n-1$ a Hamiltonian path is constructed. This path can be closed to a cycle with the edge $\{v_m+1,v_\star\}$. 
These are all ingredients to find a plane Hamiltonian cycle in the special case where all $L_i = \emptyset$. 
However in general this condition does not hold and we have to make sure to visit all vertices in $L_i$ in between. 
Figure~\ref{fig:easyexample_HC} shows an example for the case $L_i = \emptyset$.

\begin{figure}[htb]
	\centering
	\includegraphics{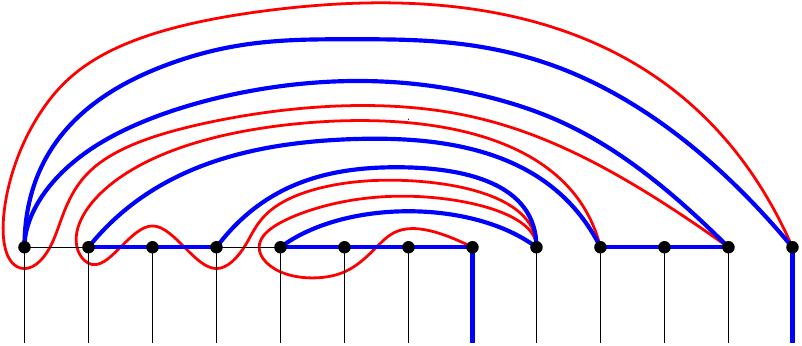}
	\caption{A plane Hamiltonian cycle (blue) in the case where $L_i = \emptyset$ for all $i$.}
	\label{fig:easyexample_HC}
\end{figure}

The strategy is to identify
edges in $E^\circ$ which connect a vertex in $L_i$ with a vertex in $R_i$ in such a way that we proceed in each step either one step to the left in $L_i$ or one step to the right in $R_i$.
These edges then allow constructing a path from $v_i$ to $v_{i+1}$.
Starting at $v_i$ we collect the remaining vertices from $L_{i-1}$ 
and continue with the vertices from $w_i^R$ to $w_i^L$ with decreasing index as in the previous case where $L_i=\emptyset$.
Additionally, we collect some of the vertices in $L_i$ until we reach one of the
chosen edges in $E^\circ$ connecting a vertex from $L_i$ to $v_i+1$, which we use.
In a second step we construct a path from $v_i+1$ to $v_{i+1}$ collecting all vertices in $R_i$ and some of the vertices in $L_i$.
This yields the desired plane Hamiltonian cycle.
The details are deferred to the appendix.

\clearpage

{
	\small
	\bibliographystyle{alphaabbrv-url}
	\bibliography{references}
}

\clearpage

\appendix

\section{SAT encoding for rotation systems}
\label{sec:encoding}

We describe a SAT encoding for simple drawings,
which is based on rotation systems.
Suppose that $\Pi=(\pi_v)_{v \in [n]}$ is the rotation system of a simple drawing of $K_n$.
Recall that $\pi_v$ is a cyclic permutation of the elements $[n] \setminus \{v\}$.

\paragraph{Global description}
We introduce a Boolean variable $X_{aib}$ for every pair of distinct vertices $a, b \in [n]$ and every index $i \in [n-1]$, to indicate whether $\pi_a(i)=b$.
To ensure that these variables indeed model a (cyclic) permutation, we assert that for every $a \in [n]$ and every $i \in [n-1]$ there is exactly one $b \in [n]$ for which $X_{aib}$ is true.
Therefore, we introduce the clauses 
\begin{itemize}
	\item 
	$\bigvee_{b:b \neq a}  X_{aib}$ for every $a \in [n]$ and $i \in [n-1]$, \quad and 
	\item $\neg X_{aib_1} \vee \neg X_{aib_2}$ for every distinct $a,b_1,b_2 \in [n]$ and $i \in [n]$.
\end{itemize}

Since we deal with cyclic permutations, we assume that the first element is always the smallest:
\begin{itemize}
	\item $X_{112}=true \quad \text{and}
	\quad   X_{k11}=true \text{ for } k \geq 2$ 
\end{itemize}

\paragraph{Local description}
Since the $X$-variables only give us a ``global'' description of the entire rotation system, we introduce additional auxiliary variables to describe the rotation system around each vertex in more detail.
We introduce an auxiliary variable $Y_{abcd}$ for every four distinct vertices $a,b,c,d \in [n]$
to indicate whether $b,c,d$ appear in counter-clockwise order in the rotation system of~$a$. 
To synchronize the $X$-variables with the $Y$-variables,
we add clauses to assert for every distinct $a,b,c,d \in [n]$ and $i,j,k \in [n-1]$ that
\begin{itemize}
	\item 
	$
	X_{aib} \wedge X_{ajc} \wedge X_{akd} 
	\rightarrow \phantom{\neg} Y_{abcd}
	$
	if $i<j<k$, $k<i<j$, or $j<k<i$,
	\quad and 
	\item
	$
	X_{aib} \wedge X_{ajc} \wedge X_{akd} 
	\rightarrow \neg Y_{abcd}
	$
	if $i<k<j$, $k<j<i$, or $j<i<k$.
\end{itemize}

Pause to note that precisely one of $Y_{abcd}$ and $Y_{abdc}$ is true, and that $Y_{abcd}=Y_{acdb}=Y_{adbc}$.

\paragraph{Forbidden subconfigurations}

Theorem~\ref{theorem:rotsys_5tuples_characterization}
asserts
that a pre-rotation system is drawable if and only if none of the obstructions $\obstructionFour$, $\obstructionFiveA$ and $\obstructionconvexFiveB$ (depicted in Figure~\ref{fig:rotsys_obstructions}) occurs as a substructure.
Since a rotation system does not contain the obstruction~$\obstructionFour$,
we derive that the two clauses
\begin{itemize}
	\item $\neg Y_{abcd} \vee
	\neg Y_{bacd} \vee
	\neg Y_{cabd} \vee
	\neg Y_{dacb}$ \quad and
	\item $\phantom{\neg} Y_{abcd} \vee
	\phantom{\neg} Y_{bacd} \vee
	\phantom{\neg} Y_{cabd} \vee
	\phantom{\neg} Y_{dacb}$
\end{itemize}

must be fulfilled for any four distinct vertices $a,b,c,d \in [n]$.
Here we ensure that neither $\obstructionFour$ nor the reversed rotation system appear as a substructure.
In an analogous manner, 
we can assert that $\obstructionFiveA$ and $\obstructionFiveB$ are not contained.
In total, we have $\Theta(n^5)$ clauses which assert that the pre-rotation system is 
indeed a rotation system, i.e., drawable.

\paragraph{Convexity and $h$-convexity}

As discussed in Section~\ref{sec:prelim},
one only needs to forbid the configurations $\obstructionconvexFiveA$ and $\obstructionconvexFiveB$ (resp.\ $\obstructionhconvexSix$)
to restrict the search space from general rotation systems to convex (resp.\ $h$-convex) drawings. 
Hence additional to the last paragraph, we need to forbid the corresponding rotation systems.

\paragraph{Crossing edges}
For $\obstructionFour$-free pre-rotation systems,
we introduce variables $C_{abcd} = C_{ef}$ to indicate whether two non-adjacent edges $e=\{a,b\}$ and $f=\{c,d\}$ cross 
(cf.\ Observation~\ref{observation:basics}(\ref{item:obstructionfour_crossingsdetermined})).
If the drawing of $K_4$ with vertices $\{a,b,c,d\}$ 
is crossing-free, the rotation system is unique up to relabelling or reflection.
Otherwise, if there is a crossing, we have one of the following: $ab$ crosses $cd$, $ac$ crosses $bd$, or $ad$ crosses $bc$ and this is fully determined by the rotation system.
For each of these three cases, there are two subcases.
If $ab$ crosses $cd$, then the directed edge $\overrightarrow{cd}$ either traverses $\overrightarrow{ab}$ from the left to the right or vice versa.
Each subcase corresponds to a unique rotation system.
The three cases are depicted in Figure~\ref{fig:rs_n4_valid}. The other cases are the reversed rotation systems which correspond to the reflected drawings.
We introduce an auxiliary variable $D_{abcd}$ to indicate whether we are in the first subcase and set $C_{abcd} = D_{abcd} \vee D_{abdc}$.

\paragraph{Plane subdrawings}

We assert that 
a subset of the edges $E' \subset E(K_n)$ 
forms a plane subdrawings
by setting the auxiliary crossing variables corresponding to pairs of edges from $E$ to $false$, that is, we have a unit-clause constraint
$
\neg C_{e,f}
$
for every pair of non-adjacent edges $e,f \in E'$.
Similarly, we can assert that $E'$ does not form a plane subdrawing with the constraint
\begin{align*}
	\bigvee_{e,f \in E' \colon e \cap f = \emptyset}
	C_{e,f}.
\end{align*}

For instance, to assert that a rotation system
does not contain a plane Hamiltonian cycle,
we need that 
for every cyclic permutation $\pi$ 
there is at least one crossing pair of edges in the set of edges $E_\pi = \{(\pi(i),\pi(i+1)) : i \in [n]\}$,
where $\pi(n+1)=\pi(1)$.
Note that this 
gives $(n-1)!$ conditions
and is therefore only suited for relatively small values of~$n$.
Similarly, we can deal with plane Hamiltonian subdrawings on $2n-3$ edges: We assert that there is at least one crossing formed by every edge set $E' \in \binom{E(K_n)}{2n-3}$ which contains the edges $E_\pi$ of some Hamiltonian cycle~$\pi$.

\paragraph{Empty Triangles}
For the definition of empty triangles in terms of rotation systems we consider it combinatorially and we distinguish the two sides of a triangle. 
We use an orientation of the sphere. 
For a triangle spanned by three distinct vertices $a,b,c$, 
one of its sides sees $a,b,c$ in clockwise order and the other one in counterclockwise order. 
We denote by  $S_{a,b,c}$ the side of the triangle which has $a,b,c$ in counterclockwise order.
In each of the drawings in Figure~\ref{fig:triangle_abc}, $S_{a,b,c}$ is the bounded side and $S_{a,c,b}$ is the unbounded side.
For every three distinct indices $a,b,c$ we introduce a Boolean variable $E_{a,b,c}$ which is $true$ if the side $S_{a,b,c}$ is empty.
To assert the $E$ variables, we introduce auxiliary variables.
For a point $d \in [n] \setminus \{a,b,c\}$, we define $E_{a,b,c}^d= false$ if $S_{a,b,c}$ contains $d$.
For four elements, there exist eight rotation systems, the four rotation systems in Figure~\ref{fig:rs_n4_valid} and their reflections.
Hence, we set the variable $E_{a,b,c}^d$ to $false$ if and only if 
one of the four cases depicted in Figure~\ref{fig:triangle_abc} occurs. 

\begin{figure}[htb]
	\centering
	\includegraphics{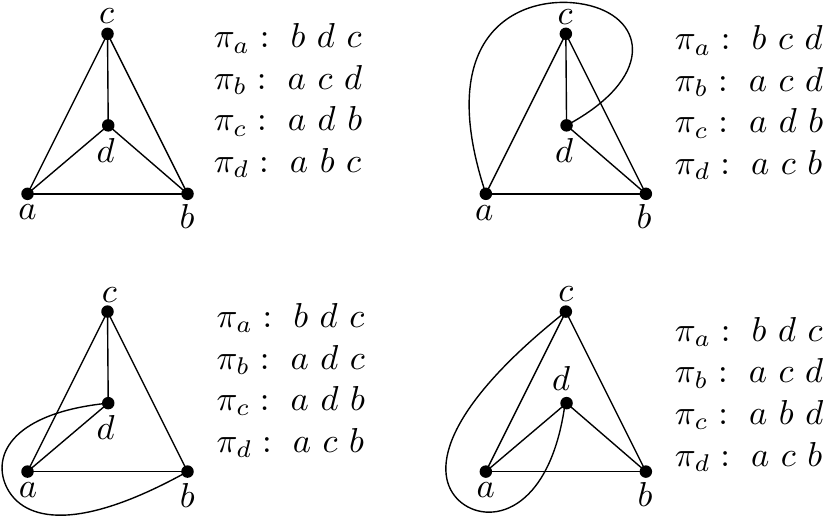}
	\caption{The four cases where $S_{a,b,c}$ contains~$d$.}
	\label{fig:triangle_abc}
\end{figure}

The side $S_{a,b,c}$ is empty if and only if
no point $d$ lies in $S_{a,b,c}$, hence 
\begin{align*}
	E_{a,b,c} = \bigwedge_{d \in [n] \setminus \{a,b,c\}} E_{a,b,c}^d.
\end{align*}

\paragraph{Symmetry breaking and
	lexicographically minimal rotation systems}

A pre-rotation system on $[n]$ is \emph{natural} if the rotation around the
first vertex is the identity permutation, that is, the vertex~$1$ sees
$2,3,\ldots,n$ in this order.  By permuting the labels of the vertices we can
make any pre-rotation system natural. Figure~\ref{fig:rs_n4_valid} shows the
four natural rotation systems of $K_4$ and their corresponding drawings.

\begin{figure}[htb]
	\centering
	\includegraphics{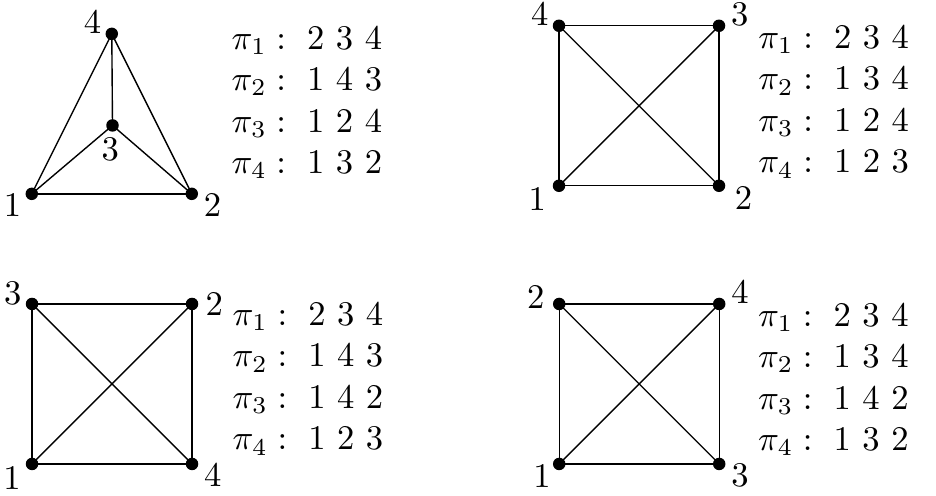}
	\caption{The four natural rotation systems on $4$ elements with their drawings. The second, third, and fourth are isomorphic and the second is the lexicographic minimum.}
	\label{fig:rs_n4_valid}
\end{figure}

For each pre-rotation system on $[n]$
we can obtain up to $2n!$ isomorphic natural pre-rotation systems 
by permuting the elements and
by reflection (which reverses all cyclic orders).
Here we describe how we restrict the SAT encoding to only consider
one representative for each isomorphism class.

A pre-rotation system~$\Pi$ can be encoded
as an $n \times (n-1)$ matrix $M_\Pi$, 
where the $i$-th row contains the permutation of the $i$-th element,
or as a vector $v_\Pi$ of length $n \cdot (n-1)$
that is obtained by concatenating the rows of the matrix one after the other.
Note that we assume that the first element in every row of $M_\Pi$ is the smallest one of the cyclic permutation, i.e., the first row starts with~$2$ and all other rows start with~$1$.
In order to find a unique representative for a relabeling class, 
we take the pre-rotation system with the lexicographically minimal vector. 

In a lexicographically minimal pre-rotation system~$\Pi$, the first row of the matrix $M_{\Pi}$ 
is the identity permutation on $[n] \backslash \{1\}$.
It is necessary that a lexicographical minimum is natural,
but not sufficient; see e.g.\ the three rotation systems on the right-hand side of Figure~\ref{fig:rs_n4_valid}.

At first glance, it seems that 
we have to check $n!$ relabellings
whenever we test whether a given pre-rotation system~$\Pi$
is a lexicographic minimum.
However, since the relabeling has to be natural,
the choice of the first and second vertices fully determines the first row and thus the full permutation of the vertices.
Hence, we only have to test $n(n-1)$ relabellings plus their reflections to test isomorphism.
There are at most $2n(n-1)$ natural pre-rotation systems in each class
and equality holds if and only if there are no symmetries.

\section{SAT encoding for planar graphs and drawing pre-rotation systems}
\label{sec:drawing}

To give an independent proof of Proposition~\ref{proposition:rotsys_classification_n6},  we describe a SAT encoding for deciding whether a given pre-rotation system is drawable.
Moreover, if a given pre-rotation system is drawable  
then every solution corresponds to 
a drawing. A rotation system has up to  $2^{\Theta({n^4})}$ different drawings with pairwise non-isomorphic 
planarizations~\cite{Kyncl2009}. 
The \emph{planarization} of a drawing is the planar graph obtained by placing an auxiliary \emph{cross-vertex} at the position of each crossing point and accordingly subdividing the crossing edges. Each cross-vertex has degree four and increases the number of edges by two.

Let $\Pi$ be a pre-rotation system on $[n]$. 
By Observation~\ref{observation:basics}(\ref{item:obstructionfour_notdrawable}) $\Pi$ is not drawable if it contains~$\obstructionFour$.
Hence we assume that $\Pi$ is $\obstructionFour$-free.
Now Observation~\ref{observation:basics}(\ref{item:obstructionfour_crossingsdetermined}) implies that
the pairs of crossing edges can be inferred from the induced 4-tuples.
We denote the set of crossing edge pairs by~$X$.
In particular, for each edge $e \in E(K_n)$ 
we know which crossings occur along it. We denote this set of crossings by~$X_e$.
The missing information to get a planarization 
(if there exists one) is the order of the crossings along the edges. 
Hence it remains to determine the order 
of the crossings along edges.

For every edge~$e = \{u,v\} \in E(K_n)$ with $u<v$,
let $\Sigma_e$ be a permutation of the crossings $X_e$ along~$e$, which we extend by adding~$u$ as the first element and~$v$ as the last element.
We define the graph $G_{\Sigma}$ where the vertex set consists of the original elements $[n]$ and the crossings~$X$, 
and for every $e=\{u,v\}\in E(K_n)$ with $u<v$,
we connect every pair of consecutive crossings from~$\Sigma_e$.

Our aim is to find an assignment for the permutations $\Sigma_e$ 
such that the corresponding graph $G_{\Sigma}$ is planar. 
For testing whether a graph is planar, 
we use Schnyder's characterization of planar graphs~\cite{Schnyder1989}, which we describe in more detail in Section~\ref{ssec:planargraphs}.

\begin{proposition}
	\label{proposition:planarization_encoding_correct}
	Let $\Pi$ be a pre-rotation system.
	If $\Pi$ is drawable, 
	then there exists an assignment $\Sigma$ such that $G_\Sigma$ is planar.
	Moreover, 
	every planar $G_\Sigma$ 
	is the planarization of a drawing with rotation system~$\Pi$.
\end{proposition}

\begin{proof}
	The first part is straight-forward.
	If $\Pi$ is drawable, we consider a drawing and obtain $G_\Sigma$ as its planarization.
	
	For the second part,
	suppose that $G_{\Sigma}$ is a planar graph
	and consider a plane drawing $D$ of~$G_{\Sigma}$.
	We show that 
	\emph{(i)} the auxiliary cross-vertices are drawn as proper crossings 
	and 
	\emph{(ii)} the rotation system $\Pi_D$ of the drawing $D$ equals $\Pi$.
	
	For the first item,
	suppose that the cross-vertex $c$ of the edges $\{u,v\}$ and $\{u',v'\}$ of $K_n$ 
	is drawn as a touching, i.e., the endpoints of the edges are consecutive in the cyclic order around $c$. Without loss of generality we assume that they appear as follows $ u',v',v,u$.
	We consider the subdrawing induced by $ u',v',v,u$.
	As illustrated in Figure~\ref{fig:propercrossing}
	not all edges of this planarization of $K_4$ can be drawn without crossings, 
	which is a contradiction.
	Hence, in~$D$ all crossing-vertices are drawn as proper crossings.
	
	Since the pairs of crossing edges are the same for $\Pi$ and $\Pi_D$, the second item
	follows from the following Proposition~\ref{prop:PRS_different_ATgraphs}.
\end{proof}

\begin{figure}[htb]
	\centering
	\includegraphics{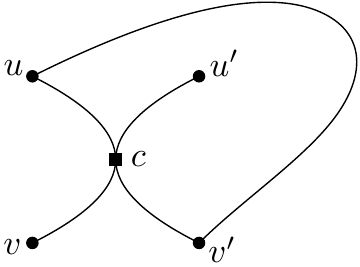}
	\caption{The vertices $u'$ and $v$ cannot be connected in a plane way.}
	\label{fig:propercrossing}
\end{figure}

\begin{proposition}
	\label{prop:PRS_different_ATgraphs}
	Let $\Pi$ and $\Pi'$ be two $\obstructionFour$-free  pre-rotation systems on $[n]$.
	The pairs of crossing edges of $\Pi$ and $\Pi'$ coincide if and only if $\Pi$ is equal to $\Pi'$ or its reflection.
\end{proposition}

A weaker version of
Proposition~\ref{prop:PRS_different_ATgraphs} (restricted to rotation systems) was already proven by Kyn\v{c}l \cite[Proposition~2.1(1)]{Kyncl2020}.

\begin{proof}
	If $\Pi$ is equal to $\Pi'$ or its reflection,
	then the pairs of crossing edges clearly coincide (Observation~\ref{observation:basics}(\ref{item:obstructionfour_crossingsdetermined})).
	
	For the converse statement, 
	assume that $\Pi$ differs from $\Pi'$ and its reflection.
	We show that we find a set of 
	five vertices $I \subset V$ such that,
	when restricted to $I$,
	the (sub-)pre-rotation system
	$\Pi|_I$ still differs from $\Pi'|_I$ and its reflection.
	
	Case~1: there exists a vertex $v$ 
	such that the cyclic permutation $\pi_v$ differs from $\pi'_v$ and its reflection.
	Then there are four elements $J = \{a,b,c,d\}$ such that,
	when restricted to $J$,  the cyclic permutation $\pi_v|_J$ still differs from $\pi'_v|_J$ and its reflection.
	Hence we can choose $I = \{v,a,b,c,d\}$.
	
	Case~2: for every vertex $v$ 
	the cyclic permutation
	$\pi_v$ is equal to $\pi'_v$ or its reflection.
	Since $\Pi$ differs from $\Pi'$ and its reflection,
	there are two vertices $v,w$
	such that $\pi_v$ equals $\pi'_v$
	and $\pi_w$ equals the reflection of~$\pi'_w$.
	By choosing any $a,b,c$ distinct from $v$ and $w$,
	we have five vertices $I=\{v,w,a,b,c\}$ 
	such that $a,b,c$ occur in the same order around~$v$
	and $a,b,c$ occur in the opposite order around~$w$.
	
	Using the program
	we enumerated all $\obstructionFour$-free pre-rotation systems on 5 vertices.
	Any two pre-rotation systems 
	which are not equal or reflections of each other
	have distinct pairs of crossing edges
	(see Appendix~\ref{sec:proof:prop:PRS_different_ATgraphs} for technical details). 
	Hence the pairs of crossing edges for $\Pi$ and $\Pi'$ are different.
\end{proof}

\subsection{A SAT encoding for planar graphs}
\label{ssec:planargraphs}

Schnyder \cite{Schnyder1989}
characterized planar graphs as graphs whose incidence poset is of order dimension at most three. 
That is,
a graph $G=(V,E)$ is planar if and only if there exist three total orders $\prec_1,\prec_2,\prec_3$ such that for every edge $\{u,v\} \in E$ and every vertex $w \in V \setminus \{u,v\}$ it holds $u \prec_i w$ and $v \prec_i w$ for some $i \in \{1,2,3\}$.
This characterization allows to encode planarity in terms of a Boolean formula.
Besides the fact that this approach can be used for planarity testing, 
it can also be used to enumerate all planar graphs with a specified property or 
to decide that no planar graphs with this property exists.
It is worth noting that 
certain subclasses of planar graphs 
such as outer-planar graphs
can be characterized in a similar fashion; see \cite{FelsnerTrotter05}.

\paragraph{The encoding}
For any two vertices $u,v \in V$ with $u \neq v$
we use Boolean variables $X_{u,v}$ 
to encode whether $\{u,v\}$ is an edge in the graph
and, for any $i \in \{1,2,3\}$, we use Boolean variables $Y_{i,u,v}$ to encode whether $u \prec_i v$.
To assert that the $Y$ variables indeed model a total order,
we must ensure the linear order and anti-symmetry with the constraints 
$Y_{i,u,v} \vee Y_{i,v,u}$ and $\neg Y_{i,u,v} \vee \neg Y_{i,v,u}$
and transitivity with the constraints 
$\neg Y_{i,u,v} \vee \neg Y_{i,v,w} \vee Y_{i,u,w}$
for every $i \in \{1,2,3\}$ and distinct $u,v,w \in V$.

To ensure Schnyder's planarity conditions,
we introduce auxiliary variables $Y_{i,u,v,w}$ to encode whether $u \prec_i w$ and $v \prec_i w$ holds for any three distinct variables $u,v,w$ and $i \in \{1,2,3\}$
(that is, $Y_{i,u,v,w} = Y_{i,u,w} \vee  Y_{i,v,w}$)
and then assert 
$
\neg X_{u,v} \vee \bigvee_{i=1}^3 Y_{i,u,v,w}.
$


\section{Characterization of rotation systems}
\label{sec:characterizationRS}

\propclassificationRS*

\begin{proof}
	We combine the SAT frameworks from Appendix~\ref{sec:encoding}
	and Appendix~\ref{sec:drawing}. 
	By Observation~\ref{observation:basics}(\ref{item:obstructionfour_notdrawable}) 
	we know that there are exactly two non-isomorphic rotation systems on 4 vertices and the non-drawable pre-rotation system $\obstructionFour$.
	Next, we enumerate the
	7 non-isomorphic pre-rotation systems on 5 vertices,
	that do not contain~$\obstructionFour$.
	The drawing framework from Appendix~\ref{sec:drawing} shows 
	that exactly five of them are drawable. See Figure~\ref{fig:computer_vis_5}
	for the computer generated drawings of the $K_5$.
	\begin{figure}[htb]
		\centering
		\includegraphics[width=0.95\textwidth]{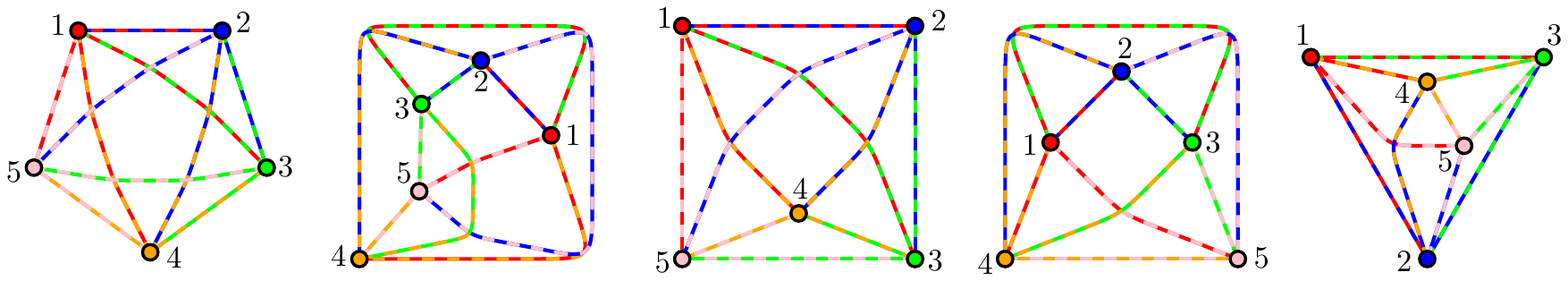}
		\caption{Computer generated drawings of the five types of rotation systems on 5 elements.}
		\label{fig:computer_vis_5}
	\end{figure}
	The
	two configurations~$\obstructionFiveA$ and~$\obstructionFiveB$ are not drawable.
	Last but not least, 
	we use the framework to enumerate 
	102 non-isomorphic pre-rotation systems on 6 elements,
	that do not contain~$\obstructionFour$, $\obstructionFiveA$, or~$\obstructionFiveB$.
	With the drawing framework we verify 
	that all of them are drawable.
	The technical details are deferred to Appendix~\ref{alternative_proof}.
\end{proof}

While a complete enumeration for $n=4$ and $n=5$ can easily be done by hand, the task for $n=6$ is complex, unpleasant, and error-prone.

\thmclassificationRS*

\begin{proof}
	Let $\Pi$ be a pre-rotation system. 
	If $\Pi$ contains at least one of the obstructions $\obstructionFour$, $\obstructionFiveA$ or $\obstructionFiveB$ as a subconfiguration, then $\Pi$ is non-drawable.
	
	Now assume that $\Pi$ is a pre-rotation system
	which does not contain the obstructions
	$\obstructionFour$, $\obstructionFiveA$, and $\obstructionFiveB$. By Proposition~\ref{proposition:rotsys_classification_n6}
	we know that all induced 4-, 5-, and 6-element subconfigurations of $\Pi$ are drawable
	and by Observation~\ref{observation:basics}(\ref{item:obstructionfour_crossingsdetermined}) the pairs of crossing edges are fully determined. 
	In particular, the pairs of crossing edges in 6-tuples belong to a drawable pre-rotation system. 
	Kyn\v cl, \cite[Theorem~1.1]{Kyncl2020} has shown that the pairs of crossing edges belong to a drawing of the complete graph if and only if this holds for all subdrawings induced by six vertices. 
	Hence, there is a drawing~$D'$ with the same pairs of crossing edges. 
	By Proposition~\ref{prop:PRS_different_ATgraphs},
	the drawing~$D'$ is a drawing of the given pre-rotation system $\Pi$ or its reflection.
	Therefore, $\Pi$ 
	is drawable.
	This completes the proof. 
\end{proof}

\section{Plane Hamiltonian subdrawings in convex drawings} 
\label{sec:convex_HC_proof}

In this section we give the full proof of
Theorem~\ref{theorem:convex_HC}. 
We prove the existence of the Hamiltonian
cycle in a constructive way and present a polynomial time algorithm which, for
a given convex drawing $D$ of the complete graph $K_n$ and a fixed vertex
$v_{\star}$, computes a plane Hamiltonian cycle which does not cross edges
incident to $v_{\star}$.  Since this statement is independent from the
labeling of the vertices, we assume throughout the section that
$v_{\star} = n$ and that the other vertices are labeled from $1$ to $n-1$ in
cyclic order around~$v_{\star}$, i.e., according to the rotation~$\pi_{n}$.
For the labels of vertices $1$ to $n-1$ we use the arithmetics modulo $n-1$. 

Since the convexity of a drawing in the plane is independent of the choice
of the outer face, we choose a drawing where $v_{\star}$ belongs to the outer face. 

We denote vertex $v_{\star}=n$ as the \emph{star vertex} and edges incident to~$v_{\star}$ as
\emph{star edges}.  Note that every Hamiltonian cycle contains exactly two
star edges.  We let $\hat{E}$ be the set of non-star edges and call an
edge $e \in \hat{E}$ \emph{star-crossing} if it crosses a star edge.
For the Hamiltonian cycle we only use edges which are not star-crossing.
Let $E^\circ=\hat{E} \setminus \{e \in \hat{E} : e \text{ is star-crossing}\}$
be the set of these candidate edges.

The following lemma follows from the properties of a simple drawing.

\lemmaHCisPlane*

Particular focus will be on the edges $e = \{v, v+1\}$ with $1\leq v < n$.
Such an edge is called \emph{good} if it is
not star-crossing. Otherwise, if
the edge $b = \{v, v+1\}$ crosses a star edge $\{w,v_{\star}\}$, then we say that $b$
is a \emph{bad edge} and $w$ is a \emph{witness} for the bad edge $b$.

If there is at most one bad
edge $\{v,v+1\}$, then the $n-2$ good edges together with the
two star edges $\{v,v_{\star}\}$ and $\{v+1,v_{\star}\}$ form a
Hamiltonian cycle which visits the non-star vertices
in the order of the rotation around $v_\star$.

\begin{restatable}{observation}{lemmaOneBadEdge}
	\label{lem:onebadedgePHC}
	If $\hat{E}$ contains at most one bad edge, then $D$ contains a
	plane Hamiltonian cycle which does not cross any star edges and visits
	the non-star vertices in the order of the rotation around the star vertex.
\end{restatable}

An example of a drawing with two bad edges is illustrated in
Figure~\ref{fig:rotsys_obstructions_hconvex} with
the choice of $v_\star=3$ and the two bad edges $\{4,6\}$ and $\{6,2\}$.

For the proof of the theorem it is critical to understand the structure
of bad edges in a convex drawing of $K_n$.
We start with an easy lemma concerning the convex side of the triangle
spanned by~$v_{\star}$ and a bad edge.
If $b = \{v, v+1\}$ is a bad edge with witness $w$, then  the side
of the triangle spanned by $\{v, v+1, v_{\star}\}$ that contains $w$ is not
convex, since $v_{\star}$ and $w$ both belong to this side but the edge
$\{v_{\star}, w\}$ is not fully contained in this side.

\obsConvexSide*

Every triangle in a convex drawing has at least one convex side. Thus, for every bad edge $\{v, v+1\}$ the triangle spanned by $\{v, v+1, v_{\star}\}$ has exactly one convex side which is the side not containing the witnesses.

The following lemma will be the key to understand the structure of bad edges.
To avoid a large case distinction, we used SAT to show that there is no convex
drawing on at most 7 vertices which violates the statement. Since convex
drawings have a hereditary structure 
(i.e., the drawing stays convex when a vertex and its adjacent edges are removed), 
this proves the lemma.

\lemmaBadEdgesNested*

\begin{proof}
	Assume towards contradiction that a convex drawing of $K_n$ contains two bad edges $b = \{v, v+1\}$ and $b' = \{v', v'+1\}$  with witnesses $w$ and $w'$, respectively, such that $w',w,v,v',$ do not occur in this particular cyclic order or reversed. 
	We may assume that $n \le 7$ because 
	the vertices $v,v+1,w,v',v'+1,w',v_{\star}$ (which are not necessarily distinct) induce a convex drawing on at most 7 vertices which has the desired property.
	Moreover, all drawings on $n \le 4$ have at most one bad edge since simple drawing on $n \le 4$ vertices have at most one crossing. 
	We used our SAT framework to enumerate all convex rotation systems on 5, 6, and 7 vertices. 
	Every solution which is a convex rotation system gets immediately tested whether they fulfill the desired property. 
	Each of them has the property that the vertices appear in the order $w',w,v,v'$ or reversed around the vertex~$n$ and it holds $w \neq w'$.
	This is a contradiction to the assumption.
	To verify this statement, run the program with the following parameters:
\begin{verbatim}
    python rotsys_find.py 5 -a -c --testNestedLemmaPart1
    python rotsys_find.py 6 -a -c --testNestedLemmaPart1
    python rotsys_find.py 7 -a -c --testNestedLemmaPart1
\end{verbatim}
	Note that the parameter \verb|-a| is
	used to enumerate and test all rotation systems
	which are then tested for the desired properties. And the parameter 
	\verb|-c| restricts the search space to convex rotation systems.
	
	To show the moreover statement, assume without loss of generality 
	that two bad edges $b$ and $b'$ with witnesses $w$ and $w'$ occur in the cyclic order $w', w, v, v+1,v',v'+1$, where possibly $v+1 = v'$ but all the other vertices are distinct.
	We show that the vertices $v',v'+1 $ are contained in the convex side $S$ of the triangle induced by $\{v,v+1,v_{\star}\}$.
	Recall that the side of the triangle containing $w$ is not convex, see lemma~\ref{obs:ConvexSide}. 
	
	If $v+1\neq v'$,
	neither of the two vertices $v'$ (Case~1) and $v' +1$ (Case~2) lies in the non-convex side of the triangle.
	To verify this statement, run the program with parameters
\begin{verbatim}
    python rotsys_find.py -c 7 --testNestedLemmaPart2Case 1 
    python rotsys_find.py -c 7 --testNestedLemmaPart2Case 2
\end{verbatim}
	
	If $v+1 = v'$, clearly $v'$ is contained in both sides of the triangle and hence also in $S$. 
	We show that $v'+1$ cannot be in the non-convex side (Case 3). 
    To verify this statement, run
\begin{verbatim}
    python rotsys_find.py -c 6 --testNestedLemmaPart2Case 3    
\end{verbatim} 
	
	Since both vertices $v'$ and $v'+1$ are contained in $S$, the bad edge $b'$ is fully contained in $S$. 
	In particular $b'$ does not cross any of the edges $b,\{w,v_{\star}\},\{v,v_{\star}\},\{v+1,v_{\star}\}$.
	Since $b'$ is a bad edge with witness $w'$, the edge $b'$ crosses the star edge $\{w',v_{\star}\}$.
	This shows that $w'$ is on one side of the triangle spanned by $b'$ and $v_{\star}$ and $v,v+1$ are on the other side. 
	With the same argument as before the side of the triangle spanned by $\{v',v'+1,v_{\star}\} $ containing $w'$ is not convex. Hence $v$ and $v+1$ are contained in the convex side of this triangle. The second part follows. 
\end{proof}

This lemma immediately implies that there is at most one bad edge
in an $h$-convex drawing.
\begin{corollary}
	Every $h$-convex drawing has at most one bad edge. 
\end{corollary}

\begin{proof}
	Towards a contradiction assume, there are two bad edges $b = \{v,v+1\}$ and $b' = \{v',v'+1\}$ with witnesses $w$ and $w'$, respectively.
	Since $h$-convex drawings are convex, we can apply \autoref{lem:twobadedges1}. 
	Hence the cyclic order of the vertices around $v_{\star}$ is either 
	$w', w, v,v'$ or the reversed order
	and $w,v,v+1$ are in the convex side of the triangle spanned by $\{v'v'+1,v_{\star}\}$ and vice versa.
	But this contradicts the property of $h$-convex drawings that the side of a triangle which is contained in a convex side of a triangle is convex. 
\end{proof}

With \autoref{lem:onebadedgePHC} this implies the moreover part of \autoref{theorem:convex_HC} for 
$h$-convex drawings.

For the case that a drawing has two or more bad edges the lemma implies that there is a partition of the vertices $1, \ldots, n-1$ into two blocks
such that in the rotation around $v_{\star}$ both blocks of the partition are consecutive, one contains the vertices of all bad edges 
and the other contains all the witnesses. 
Let $b=\{v,v+1\}$ be the bad edge whose vertices are last in the
clockwise order of its block. 
We cyclically relabel the vertices such that $b$ becomes $\{n-2,n-1\}$. This makes the labels of all witnesses smaller than
the labels of vertices of bad edges. In particular we then have the following two properties:

\begin{itemize}
	\item[\emph{(sidedness)}] If $\{v,v+1\}$ is a bad edge with
	witness $w$, then $w < v$.
	\item[\emph{(nestedness)}] If $b = \{v, v+1\}$ and $b' = \{v', v'+1\}$ are
	bad edges with respective witnesses $w$ and $w'$ and if $v < v'$, then
	$w' < w$.
\end{itemize}

In addition we can apply a change of the outer face (via stereographic
projections) such that the vertex $v_{\star}$ and the initial segments of the edges $\{v_\star,1\}$ and
$\{v_\star,n-1\}$ belong to the outer face.
The nesting property implies that we can label the bad edges 
as $b_1, \ldots, b_m$ for some $m \ge 2$, such that if $b_i = \{v_i,v_i +1\}$, then
$1 < v_1 < v_2 < \ldots < v_m = n-2$.
Moreover, we denote by $w_i^L$ and $w_i^R$ the leftmost (smallest) and the rightmost (largest) witness of the bad edge $b_i$, respectively.
Then 
$ 1 \leq w_{m}^L \leq w_{m}^R  < w_{m-1}^L \leq w_{m-1}^R < \ldots
< w_{1}^L \leq w_{1}^R  $.
Sidedness additionally implies $w_{i}^R < v_i$.
Figure~\ref{fig:layeringbadedgeswithnotation_full} shows the situation for two bad edges
$b_i$ and $b_{i+1}$.
Note that $v_i+1 = v_{i+1}$ is possible. 

\begin{figure}[htb]
	\centering
	\includegraphics[scale = 0.9]{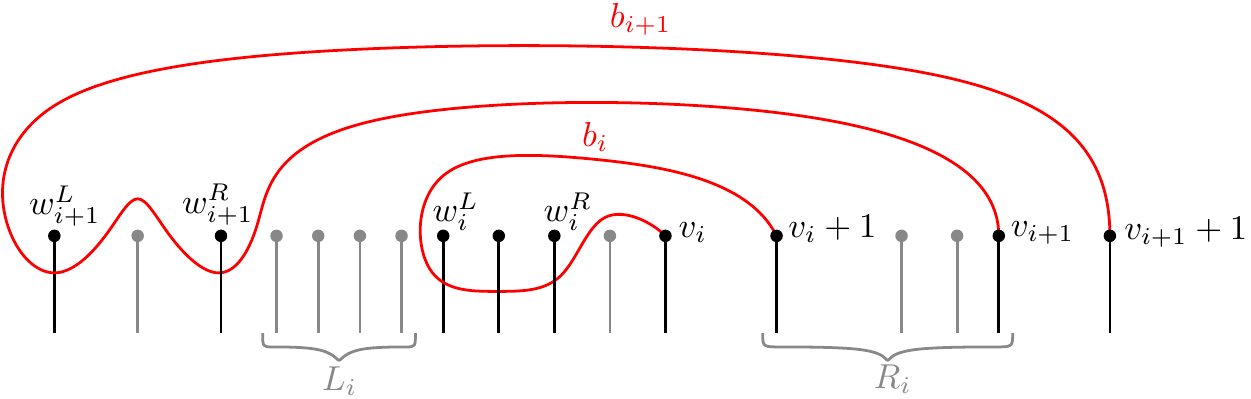}
	\caption{An illustration of sidedness and nesting for two bad edges.}
	\label{fig:layeringbadedgeswithnotation_full}
\end{figure}

Let $L_i = \{x \in [n-1]: w_{i+1}^R < x < w_i^L\}$ and
$R_i = \{x \in [n-1]: v_{i} +1\leq x \leq v_{i+1}\}$ denote the left and
the right blocks of vertices between two consecutive bad edges $b_i$ and
$b_{i+1}$, see Figure~\ref{fig:layeringbadedgeswithnotation}.
Note that $R_i$ is non-empty since it always contains $v_i+1$ and $v_{i+1}$ but $L_i$ might be empty. 

With the following lemma 
we identify useful edges in $E^\circ$.

\lemTwoUsefullEdges*

\begin{figure}[htb]
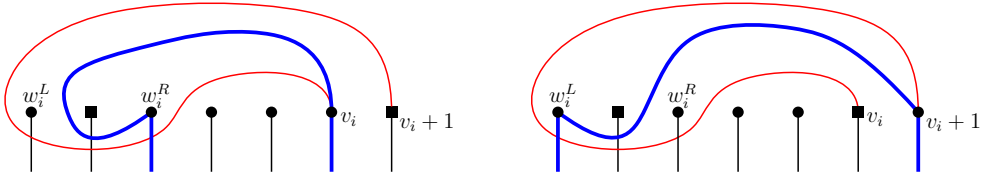

	\centering
	
		\hbox{}\hfill
		\includegraphics[page =1, scale =0.7]{figs/insidebadedge.pdf} 
		\hfill 
		\includegraphics[page =2, scale = 0.7]{figs/insidebadedge.pdf}
		\hfill\hbox{}
	
	\caption{
		The blue edges $\{w_i^R,v_i\}$, $\{w_i^L,v_i+1\}$ are not star-crossing.
		The vertices marked with the squares are witnesses for the corresponding side of the blue triangle not being convex.
	}
	\label{fig:rightmostwitness}
\end{figure}

\begin{proof}
	For a fixed $i$, both edges are fully contained in the the non-convex side of the triangle spanned by $\{v_i,v_i+1,v_{\star}\}$ since adjacent edges do not cross in a simple drawing. 
	Assume $\{w_i^L,v_i+1\}$ crosses a star edge $\{x,v_{\star}\}$. Then $x$ has to be a witness of $b_i$ with $ x > w_i^L$.
	However, the side of the triangle $\{w_i^L,v_i+1,v_{\star}\}$ that contains $x$ is not convex due to the edge $\{x,v_{\star}\}$. 
	Additionally, the other side is not convex due to the edge $\{v_i,v_i+1\}$. A similar argument holds for the edge $\{w_i^R,v_i\}$. Both situations are depicted in Figure~\ref{fig:rightmostwitness}.
\end{proof}

The general idea to  construct a plane Hamiltonian cycle
is as follows: Begin with the edge $\{v_\star,v_1\}$. When vertex $v_i$ is visited,
we go from $v_i$ to $w_{i}^R$ and then one by one with decreasing labels to $w_i^L$.
From there we want to visit
$v_i+1$ and then one by one in increasing order through the vertices of $R_i$ until we reach $v_{i+1}$. When we reach $v_m+1=n-1$ a Hamiltonian path is constructed. This path can be closed to a cycle with the edge $\{v_m+1,v_\star\}$. 
These are all ingredients to find a plane Hamiltonian cycle in the special case where all $L_i = \emptyset$. 
However in general this condition does not hold and we have to make sure to visit all vertices in $L_i$ in between. 
Figure~\ref{fig:easyexample_HC} shows an example for the case $L_i = \emptyset$.

The strategy is to identify
edges in $E^\circ$ which connect a vertex in $L_i$ with a vertex in $R_i$ in such a way that we proceed in each step either one step to the left in $L_i$ or one step to the right in $R_i$.
These edges then allow constructing a path from $v_i$ to $v_{i+1}$.
Starting at $v_i$ we collect the remaining vertices from $L_{i-1}$ 
and continue with the vertices from $w_i^R$ to $w_i^L$ with decreasing index as in the previous case where $L_i=\emptyset$.
Additionally, we continue collecting some of the vertices in $L_i$ until we reach one of the
chosen edges in $E^\circ$ connecting a vertex from $L_i$ to $v_i+1$, which we use.
In a second step we construct a path from $v_i+1$ to $v_{i+1}$ collecting all vertices in $R_i$ and some of the vertices in $L_i$.
This yields a plane Hamiltonian cycle which has
no crossing with a star edge.
The details are deferred to the appendix.

In the following we describe which edges of $E^\circ$ we use to construct the Hamiltonian cycle. 
For this we consider edges from $L_i$ to $R_i$ and analyse which edges are not star-crossing.
By the properties of a convex drawing the edges between vertices of $L_i$ and $R_i$ stay in the region between the two bad edges $b_i$ and $b_{i+1}$.
This implies that the only star edges which can be crossed by those edges are $\{v_\star,x\}$ with $x \in L_i \cup R_i$.
\begin{lemma}
	\label{lem:nocrossingsoutsidelandR}
	All edges $\{u,v\}$ with $u,v \in L_i \cup R_i$ do not cross star edges $\{z,v_{\star}\}$ with $z \in V \backslash (L_i \cup R_i)$.
\end{lemma}

\begin{proof}
	The convex sides $S_i$ and $S_{i+1}$ of the triangles spanned by the bad edges $b_i$, respectively $b_{i+1}$, and the star vertex $v_{\star}$ have a common intersection which is partitioned into three regions by the edges $\{w_{i+1}^R,v_{\star}\}$ and $\{w_i^L,v_{\star}\}$.
	Both vertices $u,v$ are contained in the region of the intersection of the convex sides $S_i$ and $S_{i+1}$ that is bounded by both edges $\{w_{i+1}^R,v_{\star}\}$ and $\{w_i^L,v_{\star}\}$.
	Since the edge $\{u,v\}$ has to lie in $S_i$ and $S_{i+1}$ and can cross $\{w_{i+1}^R,v_{\star}\}$ and $\{w_i^L,v_{\star}\}$ at most once, it has to be fully contained in the same region.
	This shows that we cannot cross star edges $\{z,v_{\star}\}$ where $z$ is outside of this region, i.e., $z \in V \backslash (L_i \cup R_i)$.
\end{proof}

Moreover, an edge from a vertex in $L_i$ to a vertex in $R_i$ cannot cross star edges $\{z,v_{\star}\}$ with $z \in R_i$.
\begin{lemma}
	\label{lem:nocrossinginRi}
	No edge $\{u,v\}$ with $u \in L_i$  and $v \in R_i$ crosses a star edge $\{z,v_{\star}
	\}$ with $z \in R_i$.
\end{lemma}

\begin{proof}
	Figure~\ref{fig:nocrossinginRi} shows the four obstructions to the statement. 
	Let $w_i, w_{i+1}$ be the witnesses of the bad edges $b_i = \{v_i,v_i+1\}$ and $b_{i+1} = \{v_{i+1}, v_{i+1}+1\}$, respectively.
	If $\{u,v\}$ crosses a star edge $\{z,v_{\star}\}$ with $z \in R_i$ such that $z >v$ we are in Case 1 or Case 2. 
	In both cases, we consider the subdrawing induced by
	$w_{i+1},u,v,z,v_{i+1}, v_{i+1}+1,v_{\star}$, which is a 
	subdrawing on 7 vertices and $u,v$ are adjacent in the rotation 
	around $v_{\star}$. Hence $\{u,v\}$ is a bad edge in the subdrawing. 
	This cannot happen in a convex drawing since the cyclic order of 
	the vertices around $v_{\star}$ violates the statement of 
	Lemma~\ref{lem:twobadedges1}.
	We argue similarly in the other two cases. 
	If $z<v$, we consider the subdrawing induced by 
	$u,w_i,v_i,v_i+1,z,v,v_{\star}$. Again the edge $\{u,v\}$ is
	a bad edge with witness $z$ since $u$ and $v$ are adjacent in the
	subdrawing. By Lemma~\ref{lem:twobadedges1} this order of vertices
	corresponding to two bad edges cannot happen.     
\end{proof}

\begin{figure}[htb]
	\centering
	
	\begin{subfigure}[b]{.45\textwidth}
		\centering
		\includegraphics[page=1,width = \textwidth]{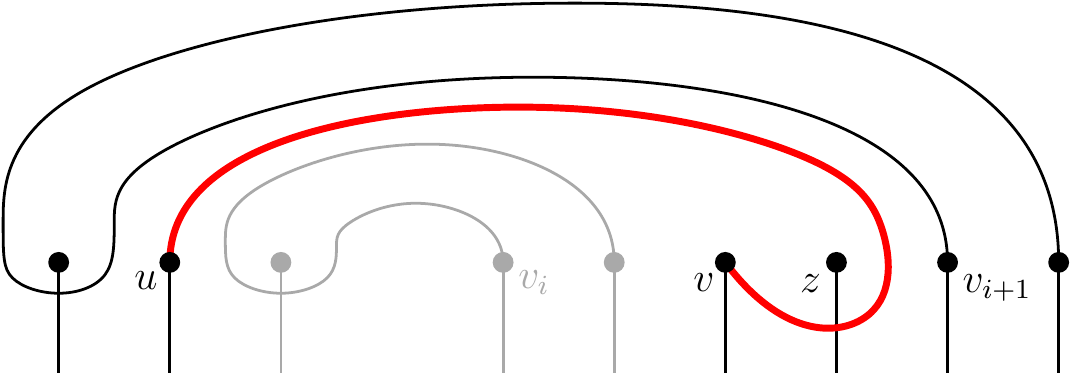}
		\caption{Case 1}
		\label{fig:nocrossinginRiCase1}
	\end{subfigure}
	\hfill
	\begin{subfigure}[b]{.45\textwidth}
		\centering
		\includegraphics[page=3,width = \textwidth]{figs/notexistingconfiguration.pdf}
		\caption{Case 2}
		\label{fig:nocrossinginRiCase2}
	\end{subfigure}
	
	\begin{subfigure}[b]{.45\textwidth}
		\centering
		\includegraphics[page=2,width = \textwidth]{figs/notexistingconfiguration.pdf}
		\caption{Case 3}
		\label{fig:nocrossinginRiCase3}
	\end{subfigure}
	\hfill
	\begin{subfigure}[b]{.45\textwidth}
		\centering
		\includegraphics[page=4,width = \textwidth]{figs/notexistingconfiguration.pdf}
		\caption{Case 4}
		\label{fig:nocrossinginRiCase4}
	\end{subfigure}
	
	\caption{An illustration of the four forbidden configurations to prove  Lemma~\ref{lem:nocrossinginRi}. The red edges cannot cross the star edge $\{z,v_{\star}\}$ as depicted.  
	}
	\label{fig:nocrossinginRi}
\end{figure}

The previous lemmas show that if $\{u, v\}$ with $u\in L_i$ and 
$v\in R_i$ intersects a star edge $\{z,v_{\star}\}$, then $z\in L_i$.
Furthermore, the edge $\{u,v\}$ cannot cross two star edges $\{z_1,v_{\star}\}$ and $\{z_2,v_{\star}\}$ with $z_1 <u <z_2$ and $z_1,z_2 \in L_i$. This is because the triangle spanned by $\{u,v,v_{\star}\}$ has at least one convex side. For an illustration, see Figure~\ref{fig:crossingsonlyoneside}.

\begin{figure}[htb]
	\centering
	\includegraphics[width = 0.45\textwidth , page = 1]{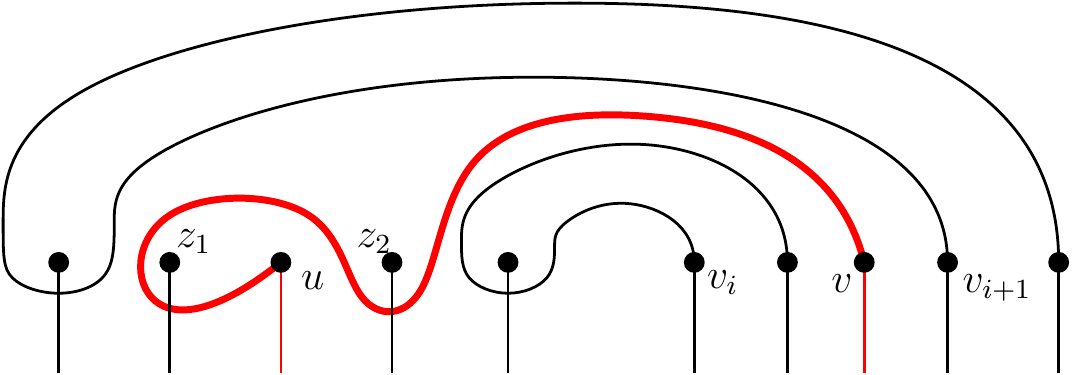}
	\hfill
	\includegraphics[width = 0.45\textwidth, page = 2]{figs/edgescrossed3.pdf}
	\caption{Illustration for the proof of Lemma~\ref{lem:crossingsonlyoneside}. 
		The red triangle $\{u,v,v_{\star}\}$ has no convex side. Witnesses for the non-convexity are the edges $\{z_1,v_{\star}\}$ and $\{z_2,v_{\star }\}$.
	}
	\label{fig:crossingsonlyoneside}
\end{figure}

\begin{lemma}
	\label{lem:crossingsonlyoneside}
	The edge $\{u,v\}$ with $u \in L_i$ and $v \in R_i$ must not cross the two star edges
	$\{z_1,v_{\star}\}$ and $\{z_2,v_{\star}\}$ with $z_1 <u < z_2$ and $z_1,z_2 \in L_i$.
\end{lemma}

However, it is possible that edges from $L_i$ to $R_i$ cross star edges of vertices in $L_i$ with either larger or smaller indices than the end vertex in $L_i$. 
Edges from $L_i$ to $R_i$ whose end-vertex in $L_i$ is the vertex with the smallest index in $L_i$ can only cross star edges $\{z, v_{\star}\}$ where $z$ is larger than the end-vertex. 
We will focus on edges which only cross star-edges with larger indices than the end-vertex in $L_i$. 
As the following lemma shows those edges help us to find edges from $L_i$ to $R_i$ which are not star crossing.

\begin{lemma}\label{lem:edgesfromlefttoright}
	Let $u \in L_i$ and $v \in R_i$ and $z$ be the largest integer in $L_i$ with $z>u$
	such that the edge $\{u,v\}$ crosses the star edge $\{z,v_{\star}\}$.
	Then the following two statements hold:
	\begin{enumerate}[(a)]
		\item 
		\label{lem:edgesfromlefttoright:item1}
		The edge $\{z,v\}$ is not star-crossing.
		\item 
		\label{lem:edgesfromlefttoright:item2}
		The edge $\{z+1,v\}$ does not cross the edges $\{x',v_{\star}\}$ with $x' \in L_i$ and $x' \le z$. 
	\end{enumerate}
\end{lemma}

\begin{proof}
	To show~(\ref{lem:edgesfromlefttoright:item1}),
	assume towards a contradiction that the edge $\{z,v\}$  crosses a star edge $\{x',v_{\star}\}$.
	From Lemma~\ref{lem:nocrossingsoutsidelandR} and Lemma~\ref{lem:nocrossinginRi} we know that $x' \in L_i$.
	Moreover, the edge $\{z,v\}$ has no crossing with the 
	triangle $\triangle$ induced by $\{u,v,v_{\star}\}$.
	Hence $\{z,v\}$ does not cross any star edge which is fully contained in the side of $\triangle$ which does not contain $z$.
	The choice of $z$ implies that all edges $\{x,v_{\star}\}$ 
	with $x \in L_i$ and $x > z $ do not cross $\{u,v\}$.
	Hence  $\{v,z\}$ can only cross star edges 
	$\{x',v_{\star}\}$ with $x' \in L_i$ and $x' < z$.
	
	Now observe that the triangle induced by $\{z,v,v_{\star}\}$ is not convex. The side containing the vertex $x'$ is not convex since the edge $\{x',v_{\star}\}$ crosses $\{z,v\}$.
	The other side contains the vertex $u$ and is not convex since the edge $\{u,v\}$ crosses $\{z,v_{\star}\}$.
	Figure~\ref{fig:edgesfromlefttoright} gives an illustration
	(the figure shows the two cases $x'>u$ and $x'<u$, which we do not need to distinguish).
	This is a contradiction. 
	
	\begin{figure}[htb]
		\centering
		
		\includegraphics[scale =0.7]{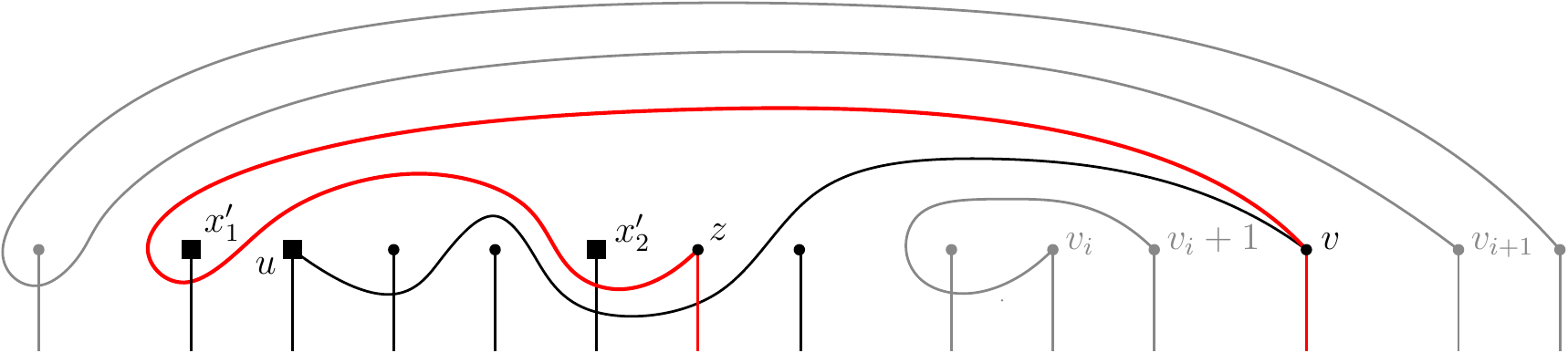}
		
		\caption{Illustration for the proof of Lemma~\ref{lem:edgesfromlefttoright}(\ref{lem:edgesfromlefttoright:item1}). The red triangle $\{z,v,v_{\star}\}$ has no convex side. Witnesses for the non-convexity are the edges $\{u,v_{\star}\}$ and $\{x_i',v_{\star }\}$.
		}
		\label{fig:edgesfromlefttoright}
	\end{figure}
	
	To show~(\ref{lem:edgesfromlefttoright:item2}),
	assume towards a contradiction that $\{z+1,v\}$ crosses a star edge $\{x',v_{\star}\}$ with $x' \in L_i$ and $x' \leq z$. 
	From~(\ref{lem:edgesfromlefttoright:item1}) we know 
	that $\{z,v\}$ does not cross any star edges.  
	Since $\{z,z+1\}$ is a good edge, the vertices $u$ and $v_\star$ are contained in different sides of the triangle $\{z,z+1,v\}$.
	This is a contradiction since both sides are not convex. 
	Witnesses for the non-convexity are the star vertex~$v_{\star}$ whose edge $\{z,v_{\star}\}$ crosses the boundary of the triangle and the edge $\{u,v\}$.
	An illustration is given in Figure~\ref{fig:edgesfromlefttoright2}.
\end{proof}

\begin{figure}[htb]
	\centering
	\includegraphics[scale = 0.7]{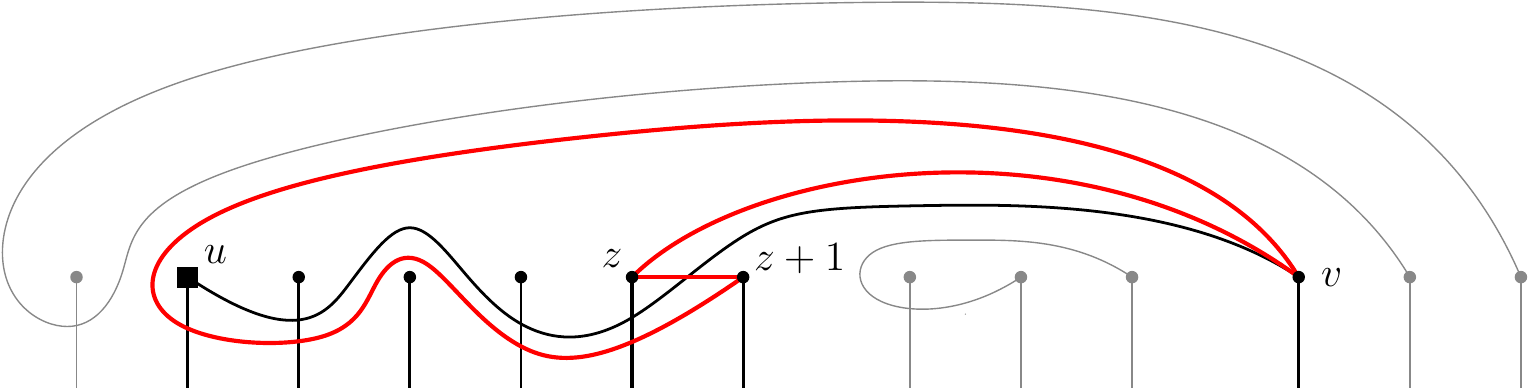}
	\caption{Illustration for the proof of Lemma~\ref{lem:edgesfromlefttoright}(\ref{lem:edgesfromlefttoright:item2}). The red triangle has no convex side. Witnesses for the non-convexity are the edges $\{z,v_\star\}$ and $\{u,v\}$.
	}
	\label{fig:edgesfromlefttoright2}
\end{figure}

If $z \in L_i$ is the vertex with the largest index such that the star edge $\{z,v_{\star}\}$ is crossed by an edge from some $x \in L_i$ to $r \in R_i$, then the edge $\{z+1,r\}$ is not star-crossing. 
For this we define the vertex $l(r) \in L_i \cup \{-\infty\}$ for all $r \in R_i$ recursively.
As a starting value we set $l(v_i) = w_i^L$.
For $r \in R_i$, 
$l(r)$ denotes the largest index in $L_i$ such that the star edge $\{l(r),v_{\star} \}$ is crossed by an edge from $L_i$ to $r$ such that the end vertex in $L_i$ has a smaller index than $l(r-1)$.
More formally, it is
\begin{align*}
	l(r) = \max \{l \in L_i: \text{ edge } \{l,v_{\star}\} \text{ crosses } \{l', r\} \text{ with } l' \in L_i\text{ , } l' <l \text{ and } l' < l(r-1) \},
\end{align*}

Note that a vertex $l \in L_i$ with the desired properties does not necessarily exist, in which case we have $l(r) = -\infty$.
In the case where one of the $l(r) = - \infty$, we can easily construct a path from $l(r-1) \in L_i$ to $v_{i+1}$ as follows:
Since $l(r) = -\infty$, the edge $\{w_{i+1}^R +1, r\}$ is not star-crossing. Hence we can go from $l(r-1)$  
in $L_i$ to the vertex with the smallest index in $L_i$ which is $w_{i+1}^R +1$, take the edge $\{w_{i+1}^R +1, r\}$ to go to $R_i$ and collect the remaining vertices one by one starting from $r$ with increasing index to $v_{i+1}$.

In general we can show that the $l(r)$ have decreasing indices. 

\begin{lemma} \label{lem:l(vi)}
	For $r,r' \in R_i$ with $r<r'$, 
	we have $l(r) \geq l(r')$. 
\end{lemma}

\begin{figure}[htb]
	\centering
	\includegraphics[scale =0.7]{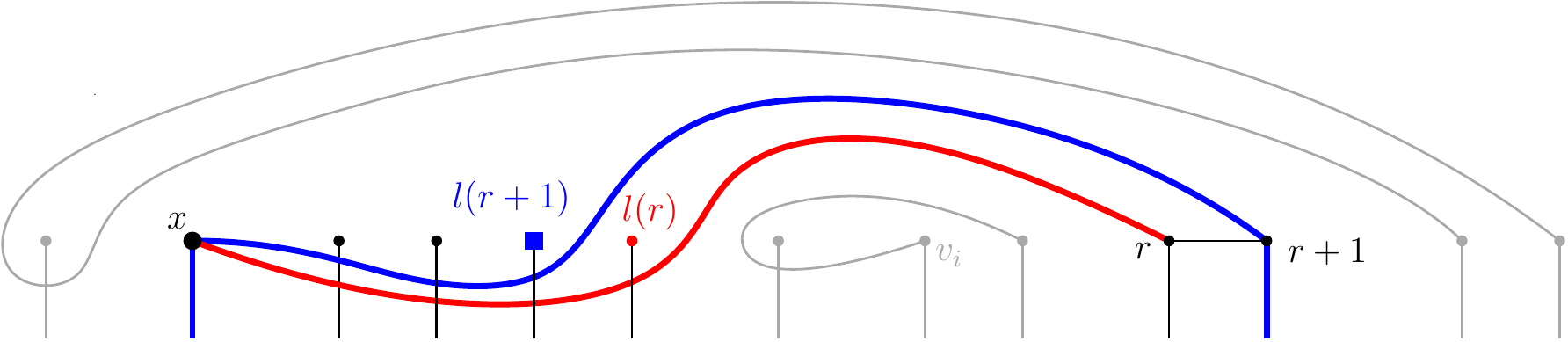}
	\caption{Illustration of the proof of Lemma~\ref{lem:l(vi)}. }
	\label{fig:lefttoright}
\end{figure}

\begin{proof}
	We only consider the case $r' = r+1$. 
	For $r' > r+1$ the claim follows by transitivity. 
	Let $\triangle$ be the triangle induced by the vertices $\{ v_{\star}, r+1, x\} $ such that $\{r+1,x\}$ crosses $\{l(r+1),v_{\star}\}$. See the blue triangle in Figure~\ref{fig:lefttoright}.
	The side of $\triangle$ containing $l(r+1)$ is not convex. 
	Hence the other side has to be convex. 
	Both vertices $x$ and $r$ are contained in this convex side. 
	Hence the edge $\{x,r\}$ is fully contained in $\triangle$ and hence it crosses the star edge $\{ l(r+1),v_{\star}\}$.
	Hence $l(r) \geq l(r+1)$.
\end{proof}

Together with Lemma~\ref{lem:edgesfromlefttoright}
it follows:
\begin{lemma}\label{cor:edgeslefttoright}
	Let $r \in R_i$ with $l(r) \neq -\infty $.
	Then the edge $\{l(r),r\}$ is not star-crossing.
	Moreover, the edge $\{l(r)+1,r\}$ is not star-crossing if $r = v_{i+1}$ or 
	$l(r) \neq l(r-1)$.
\end{lemma}

\begin{proof}
	For $r \in R_i$ with $l(r) \neq -\infty$, the edge $\{l(r),r\}$ is not star-crossing by Lemma~\ref{lem:edgesfromlefttoright}(\ref{lem:edgesfromlefttoright:item1}).
	
	Let us now focus on the edge $\{l(r)+1,r\}$.
	If $r = v_{i+1}$, the edge $\{l(v_{i+1})+1,v_{i+1}\}$ with $l(v_{i+1}) \neq \infty$ is not star-crossing. 
	In the case $l(v_{i+1})+1 = w_i^L$, this follows from Lemma~\ref{lem:leftrightmostwitness}.
	Otherwise, $l(v_{i+1})+1 < w_i^L$, and it follows by the maximality of $l(v_{i+1})$.
	
	To show the last part of the statement, assume $r \in R_i \backslash \{v_{i+1}\}$ and $l(r)\neq l(r-1)$.
	By Lemma~\ref{lem:edgesfromlefttoright} we know $l(r) < l(r-1)$.
	By Lemma~\ref{lem:nocrossinginRi}, the edge $\{l(r)+1,r\}$ does not cross star edges $\{x,v_{\star}\}$ with $x\in R_i$ 
	and by Lemma~\ref{lem:edgesfromlefttoright}(\ref{lem:edgesfromlefttoright:item2}) we know that the edge $\{l(r)+1,r\}$ does not cross star edges $\{x,v_{\star}\}$ with $w_{i+1}^R < x \leq l(r)$.
	Hence it remains to show that there are no crossings with $\{x,v_{\star}\}$ for $l(r)+1 \leq x \leq w_i^L$.
	For this we consider two cases. 
	
	First let $l(r)+1 < l(r-1)$.
	When determining the vertex $l(r)$, we look at all vertices $l'< l(r-1)$. 
	Hence the edge $\{l(r)+1,r\}$ does not cross star edges $\{x,v_{\star}\}$ with $l(r)+1 \leq x < w_i$ by the maximality of $l(r)$. 
	
	In the remaining case, it is $l(r)+1 = l(r-1)$.
	In this case the edge cannot cross star edges since it cannot cross the adjacent edge witnessing the value of $l(r)$ and it cannot cross the edge witnessing the value of $l(r-1)$ twice.
	For an illustration, see Figure~\ref{fig:lefttoright2}.
\end{proof}

\begin{figure}[htb]
	\centering
	\includegraphics[scale =0.7]{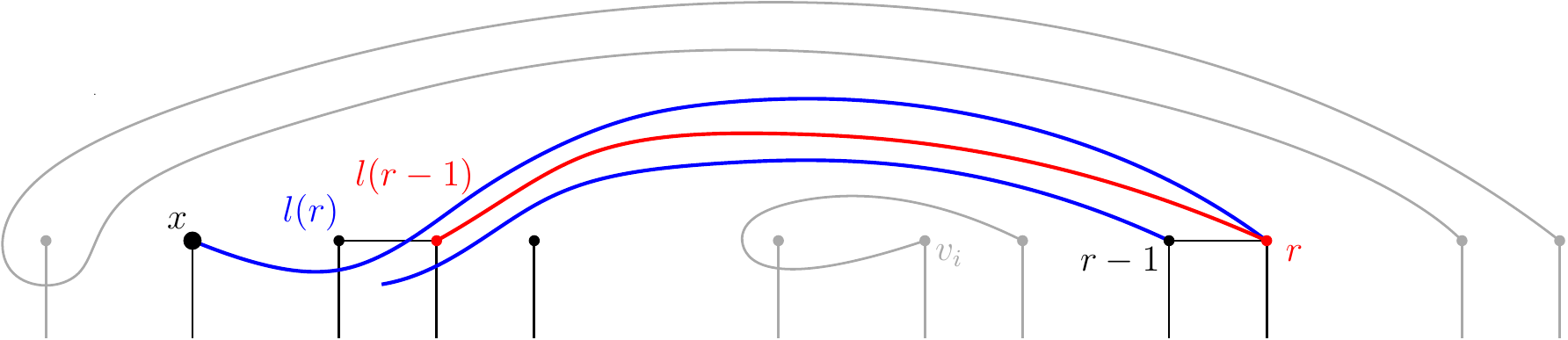}
	\caption{Illustration of the second part of~\autoref{cor:edgeslefttoright}.}
	\label{fig:lefttoright2}
\end{figure}

Similar to the case with $L_i = \emptyset$,
we construct a Hamiltonian cycle starting with the edge $\{v_1,v_{\star}\}$.
For $i = 1, \ldots , m-1$ we successively add paths from $v_i$ to $v_{i+1}$.
In the end we use good edges to traverse the remaining vertices
and close the cycle with the edge $\{v_m+1,v_{\star}\}$.

The path between $v_i$ and $v_{i+1}$ is
divided into two subpaths. 
The first part is a path from $v_i$ to $v_i+1$ and the second part from $v_i+1$ to $v_{i+1}$.
Starting from $v_i$, we go to the  unvisited vertex in $L_{i-1}$ with the largest index. 
Note that for $i =1$, we set $L_0 = [w_1^R+1,v_1-1]$.
Then we visit them one by one with decreasing index until we reach the witness $w_i^R$. We proceed visiting vertices one by one with decreasing index until we reach $w_i^L$.
Now we compute $l(v_{i}+1)$.
In the case $l(v_{i}+1) = - \infty$, we proceed as described before and continue collecting all vertices in $L_i$ one by one with decreasing index and then use the edge $\{w_{i+1}^R+1,v_{i}+1\} $ to go to $R_i$.
In this case the second part of the path only consists of collecting the remaining vertices in $R_i$ with increasing index until we reach $v_{i+1}$.

If $l(v_{i+1}) > - \infty$, we proceed from $w_i^L$ to $l(v_{i+1})+1$ one by one with decreasing index. 
By \autoref{cor:edgeslefttoright} the edge $\{l(v_{i}+1)+1, v_i +1\}$ is not star crossing. We use this edge to get to $v_i +1$.
Now we continue with the second part of the path to reach $v_{i+1}$.
Let $r< v_{i+1}$ be the current vertex in $R_i$. 
We compute $l(r+1)$.
If $l(r+1) = l(r)$, we use the good edge $\{r,r+1\}$ and continue with $r+1$ as current vertex in $R_i$.
If $l(r+1) = - \infty$, we use the edge 
$\{r,l(r)\}$ which is not star crossing by \autoref{cor:edgeslefttoright}.
From $l(r)$ we continue collecting the remaining vertices with smaller index in $L_i$ until $w_{i+1}^R+1$, using the edge $\{w_{i+1}^R+1, r+1\}$ and continue collecting the remaining vertices in $R_i$ with increasing index. 

In the remaining case, i.e., $- \infty < l(r+1) <l(r)$, we use the edge $\{r,l(r)\}$ to go to~$L_i$, continue collecting the vertices in $L_i$ one by one with decreasing index until we reach $l(r+1)+1$.
Then we use the edge $\{l(r+1), r+1\}$ which is not star crossing by \autoref{cor:edgeslefttoright}.
Now we proceed with $r+1$ as the current vertex in $R_i$.

For the last bad edge $b_m$, we construct the path to $v_{\star}$ as follows:
Use the edge $v_m$ to the largest vertex in $L_{m-1}$ which was not used yet. 
Then proceed with decreasing index to reach vertex $1$, and then add the good edge $\{1,v_m+1\}$. 
Finally, we only have to add the star edge $\{v_m+1,v_{\star}\}$ to close the path to a cycle.

For a formal description of the algorithm, see Algorithm \ref{alg:planHC}.

\begin{algorithm}[!htb] 
	\caption{Construction of a plane Hamiltonian cycle for convex drawings}\label{alg:planHC}
	\begin{algorithmic}[1]
		\State Start the cycle with edge $\{v_{\star}, v_1\}$
		\State Set $u_1 = v_1 -1$
		\For{$i = 1, \ldots, m-1$}
		\State Add edge $\{v_i, u_i\}$ 
		\label{alg:planHC:line4}
		\State Set $r = v_i+1$
		\While{$r \leq v_{i+1}$}
		\If{$l(r)=-\infty$} 
		\State Add edges $\{u_i,u_i-1\}, \ldots, \{w_{i+1}^R+2,w_{i+1}^R+1\}$ 
		\label{alg:planHC:line8good}
		\State Add edge $\{w_{i+1}^R+1,r\}$
		\label{alg:planHC:line9}
		\State Add edges $\{r,r+1\}, \ldots, \{v_{i+1}-1,v_{i+1}\}$ 
		\label{alg:planHC:line10good}
		\State Set $u_{i+1} = w_{i+1}^R$
		\State Set $r = +\infty$
		\ElsIf{$l(r) \in L_i$} 
		\State Add edges 
		$\{u_i,u_i-1\}, \ldots, \{l(r) +2,l(r) +1\}$
		\label{alg:planHC:line12good}
		\State Add edge $\{l(r) +1,r\}$ 
		\label{alg:planHC:line13}
		\State Let $r' \in R_i$ be the largest index such that $l(r)  = l(r+1) = \ldots = l(r')$ 
		\State Add edges $\{r,r+1\},\{r+1,r+2\}, \ldots,\{r'-1,r'\}$
		\label{alg:planHC:line15good}
		\If{$r' < v_{i+1}$}
		\State Add edge $\{r',l(r')\}$ 
		\label{alg:planHC:line19}
		\State Set $u_i = l(r')$
		\ElsIf{$r' = v_{i+1}$}
		\State Set $u_{i+1} = l(v_{i+1})$
		\EndIf
		\State Set $r = r'+1$
		\EndIf
		\EndWhile
		\EndFor
		\State Add edge $\{v_m, u_m\}$  
		\label{alg:planHC:line25}
		\State Add edges  $\{u_m,u_m-1\}, \{u_m-1,u_m-2\}, \ldots, \{2,1\}$ 
		\label{alg:planHC:line26good}
		\State Add edge $\{1,v_m+1\}$ 
		\label{alg:planHC:line27good} 
		\State Add star edge $\{v_m+1,v_{\star}\}$
	\end{algorithmic}

\end{algorithm}


\begin{lemma}
	\label{lemma:cycle_noncrossing}
	Algorithm~\ref{alg:planHC} finds a plane Hamiltonian cycle 
	whose edges are not star-crossing in~$O(n^2)$ time.
\end{lemma}

\begin{proof}
	In each step of Algorithm~\ref{alg:planHC}  
	we either progress to the left or to the right side without skipping vertices.
	Hence, every vertex is visited exactly once
	and the resulting cycle is Hamiltonian.
	
	Moreover, since all the edges of the constructed cycle do not cross star edges and we construct the edges in such a way that the vertices of two independent edges $e =\{u,v\}$ and $e' = \{u',v'\}$ appear cyclically in the order $u,u',v',v$, they fulfill the condition of~\autoref{lem:HCplane}.
	Hence the constructed Hamiltonian cycle is plane. 
	
	To see that no star-crossing edges are added,
	we analyze which edges are added by the algorithm:
	\begin{itemize}
		\item 
		All edges that are added
		in lines~\ref{alg:planHC:line8good}, 
		\ref{alg:planHC:line10good},
		\ref{alg:planHC:line12good},
		\ref{alg:planHC:line15good},
		\ref{alg:planHC:line26good}, and
		\ref{alg:planHC:line27good}
		are between two consecutive vertices where both vertices lie in $L_i$ or $R_i$ 
		and hence they are good edges, i.e., 
		they do not cross star edges. 
		\item In lines~\ref{alg:planHC:line4} and~\ref{alg:planHC:line25} we add $\{v_i,u_i\}$:
		For $i =1$, $\{v_1,u_1\}$ is a good edge because $u_1 = v_1-1$. 
		For $i>1$, we have either $u_i = w_i^R$ or $u_i = l(v_i)$.
		In the first case the edge is not star crossing by~\autoref{lem:leftrightmostwitness} and~\autoref{cor:edgeslefttoright} implies that in the second case, the edge does not cross star edges.
		\item In line~\ref{alg:planHC:line9} we add $\{w_{i+1}^R + 1,r\}$:
		If we add an edge of this form, we are in the case where $l(r) = -\infty$.
		This implies that the edge $\{w_{i+1}^R + 1,r\}$ does not cross star edges. 
		\item In line~\ref{alg:planHC:line13}  we add $\{l(r) + 1,r\} $:
		In this case $l(r) \in [n-1]$ and by~\autoref{cor:edgeslefttoright} this edge does not cross star edges. 
		\item In line~\ref{alg:planHC:line19}  we add $\{r', l(r')\}$:
		By~\autoref{cor:edgeslefttoright} this edge does not cross star edges. 
	\end{itemize}
	This shows the correctness of the Algorithm. 
	
	\paragraph{Running time:}
	In the first preprocessing step, we compute the bad edges. This is possible in $O(n^2)$ time:
	for each of the $n-1$ edges $\{i,i+1\}$
	there are $n-3$ potential witnesses which have to be tested.
	This directly determines the value $m$ and 
	all values $v_i, w_i^L$, $w_i^R$.
	
	In the second preprocessing step, we compute the values of $l(r)$ for every $r$.
	We claim that this can be done in $O(n^2)$ time.
	To determine $l(r)$ for $r \in R_i$, we check if $l = l(r)$ for every $l \in L_i$.
	Recall that $l(r+1) \leq l(r)$ due to~\autoref{lem:l(vi)}.
	Thus, once $l(r)$ is determined, we only need to check for vertices $l \in L_i$ where $l \leq l(r)$ to determine $l(r+1)$.
	We start with the smallest index $r$ from $R_1$ and with the largest index $l$ from $L_1$, and iteratively
	either decrement $l$ by one if $l \neq l(r)$, or we find $l = l(r)$ and increment $r$ by one.
	In total, we consider a linear number of candidate pairs $(l,r)$ to determine all values $l(r)$.
	For each such pair $(l,r)$, 
	we can test in linear time whether $l=l(r)$ holds by checking whether the edge $\{l,v_\star\}$ crosses some edge $\{l',r\}$ with $l' < l$. Altogether, we can compute all values of $l(r)$ in $O(n^2)$ time.
	
	Next observe that
	in each loop at least one edge is added to the cycle. 
	Since in total there are exactly $n$ edges added to the cycle, 
	there are at most $n$ iterations of the loop which takes $O(n)$ time.
	Hence, the total running time of Algorithm~\ref{alg:planHC} is $O(n^2)$.
\end{proof}

This completes the proof of Theorem~\ref{theorem:convex_HC}.

\goodbreak

\section{Hamiltonian paths with prescribed edges}
\label{app:HP_prescribed_edges}

In this section we derive Theorem~\ref{thm:HP_prescribed_edges} from Theorem~\ref{theorem:convex_HC}. 

Let $e = \{u,v\}$ be an edge in a convex drawing~$D$.
By Theorem~\ref{theorem:convex_HC} 
there exists a plane Hamiltonian subgraph
containing all edges adjacent to~$u$. 
Assume that the Hamiltonian cycle traverses $u,x_1,\ldots,x_{n-1}$ in this order with $v=x_i$ for some index $i$.
If $i=1$ or $i=n-1$, the Hamiltonian cycle contains the edge $\{u,v\}$ and therefore fulfills the desired properties.
Otherwise 
we start at $v=x_{i-1}$, 
traverse $x_{i-2},\ldots,x_{1}$,
take the star edges $\{x_1,u\}$ and $e = \{u,x_i\}$,
and finally traverse $x_{i+1},\ldots,x_{n-1}$.
Since all traversed edges are part of 
the original plane Hamiltonian subgraph,
the proposed Hamiltonian path is also plane.

\section{Further Description and Technical Details}
\label{sec:technicaldetails}

We provide the source code and many examples as supplemental data \cite{supplemental_data}.

The python program \verb|rotsys_find.py| (which plays a central role in this article) comes with a mandatory parameter \verb|n| for the number of vertices and several optional parameters.
The program creates a SAT instance
and then uses the Python interfaces \verb|pycosat| \cite{pycosat}  and  PySAT \cite{pysat} to run the
SAT solver PicoSAT, version 965, \cite{Biere08} and CaDiCaL, version~1.0.3 \cite{Biere2019}, respectively.
As default, we use the solver CaDiCaL because it is more efficient, and whenever all solutions need to be enumerated, we use PicoSAT.  
The output of the program is then either a rotation system with the desired properties 
(which is obtained from parsing the variable assignment of a solution to the instance) 
or, if the instance is unsatisfiable, it prints that no solution exists.
In the case the instance has various solutions,
one can use \verb|-a| to enumerate all solutions.

In the case the instance is unsatisfiable, 
we can ask the solver for a certificate and use the independent proof checking tool DRAT-trim~\cite{WetzlerHeuleHunt14} to verify its correctness.
More specifically,
we can use the parameter \verb|-c2f| to export the instance to a specified file, which has the DIMACS CNF file format.
By running
\begin{verbatim}
    cadical -q --unsat instance.cnf instance.proof
\end{verbatim}
CaDiCaL exports a DRAT proof which certificates the unsatisfiability.
The correctness of the DRAT proof can then be verified with 
\begin{verbatim}
    drat-trim instance.cnf instance.proof -t 999999
\end{verbatim}
Note that the parameter \verb|-t| increases the time limit for DRAT-trim,
which is quite low by default.

In the following, 
we will only describe how to use the program \verb|rotsys_find.py|
to show certain statements and
omit the explicit commands to certifying unsatisfiability.

\subsection{Proof of Proposition~\ref{prop:PRS_different_ATgraphs}}
\label{sec:proof:prop:PRS_different_ATgraphs}

To enumerate all $\obstructionFour$-free pre-rotation systems on 5 vertices
and test that any two
that are not obtained via reflection
have distinct pairs of crossing edges,
run the following command:
\begin{verbatim}
    python rotsys_find.py 5 -v5 -a -nat -cpc
\end{verbatim}
Here the parameters have the following purposes:
\begin{itemize}
	
	\item 
	The first parameter specifies the number of vertices $n=5$.
	
	\item
	By default, the program excludes the three obstruction $\obstructionFour$, $\obstructionconvexFiveA$ and $\obstructionconvexFiveB$.
	The parameter \verb|-v5| specifies that 
	the obstructions $\obstructionconvexFiveA$ and $\obstructionconvexFiveB$ are not excluded
	(only $\obstructionFour$ is excluded).
	Hence, the solutions are all $\obstructionFour$-free pre-rotation systems.
	
	\item
	The parameter \verb|-a| specifies that 
	all solutions should be enumerated.
	
	\item
	By default, the program only searches for (pre-)rotation systems with natural labeling.
	The parameter \verb|-nat| specifies that 
	we also search for solutions which are not naturally labeled.
	
	\item
	The parameter \verb|-cpc| specifies that 
	the pairs of crossing edges should be checked:
	any two rotation systems (not obtained via reflection)
	have distinct crossing pairs. 
\end{itemize}

\subsection{Alternative proof of Proposition~\ref{proposition:rotsys_classification_n6}}
\label{alternative_proof}

The SAT frameworks from Section~\ref{sec:encoding}
and from Section~\ref{sec:drawing}
can be combined to verify the correctness of Proposition~\ref{proposition:rotsys_classification_n6},
which classifies drawable pre-rotation systems on 4, 5, and 6 vertices from \'Abrego et al.\ \cite{AbregoAFHOORSV2015}.

As a first step, we use
the SAT framework to enumerate the
3 non-isomorphic pre-rotation systems on 4 elements:
\begin{verbatim}
    python rotsys_find.py -a -l -v4 4
\end{verbatim}
Here the parameters are as follows:
\begin{itemize}
	\item \verb|-v4| do not exclude the obstruction $\obstructionFour$ as a subconfiguration;
	\item \verb|-a| enumerate all solutions;
	\item \verb|-l| only enumerate lexicographically minimal (pre-)rotation systems.
\end{itemize}
It is well known (a simple case distinction shows)
that $K_4$ has exactly two non-isomorphic drawings with 0 and 1 crossings, respectively,
which are illustrated in Figure~\ref{fig:rs_n4_valid}.
Therefore exactly the two corresponding rotation systems are realizable.
The third pre-rotation system, which is the obstruction~$\obstructionFour$ depicted in Figure~\ref{fig:rotsys_obstructions}, is not drawable.

Next, we use the framework to enumerate 
7 non-isomorphic pre-rotation systems on 5 elements,
which do not contain~$\obstructionFour$:
\begin{verbatim}
    python rotsys_find.py -a -l -v5 5 -r2f all5.json0
\end{verbatim}
Here the parameter \verb|-v5| is used to not exclude the obstructions $\obstructionFiveA$ and $\obstructionFiveB$ as subconfigurations,
and \verb|-r2f| is used to export all solutions to the specified file. 
Then we use the drawing framework to verify 
that five of them are drawable (see Figure~\ref{fig:computer_vis_5})
and that the
two configurations to~$\obstructionFiveA$ and~$\obstructionFiveB$ (see Figure~\ref{fig:rotsys_obstructions}) are non-drawable:
\begin{verbatim}
    sage rotsys_draw.sage all5.json0
\end{verbatim}
We ran the program \verb|rotsys_draw.sage| 
with SageMath version 9.4 \cite{sagemath_website}.
The first parameter specifies the input file,
where each line encodes a rotation system.
To visualize the computed planarizations 
and create image files,
one can use the optional \verb|-v| parameter.

The case distinction needed to show that $\obstructionFiveA$
and~$\obstructionFiveB$ are non-drawable can also be done by hand.

Finally, 
we use the framework to enumerate 
102 non-isomorphic pre-rotation systems on 6 elements,
which do not contain~$\obstructionFour$, $\obstructionFiveA$, and~$\obstructionFiveB$:
\begin{verbatim}
    python rotsys_find.py -a -l 6 -r2f all6.json0
\end{verbatim}
Then we use the drawing framework to verify 
that all of them are drawable:
\begin{verbatim}
    sage rotsys_draw.sage all6.json0
\end{verbatim}
This completes the proof of Proposition~\ref{proposition:rotsys_classification_n6}.

\subsection{Plane Hamiltonian substructures}
\label{plane_substructures}

\paragraph{Rafla's conjecture (Conjecture~\ref{conjecture:rafla})}

To verify that Conjecture~\ref{conjecture:rafla} holds for $n \le 10$,  
we used the Python program with the following parameters:
\begin{verbatim}
    for n in {3..10}; do 
        python rotsys_find.py -HC $n
    done
\end{verbatim}
The parameter \verb| -HC| asserts that there exists no plane Hamiltonian cycle.
Using the solver CaDiCaL,
it took about 6 CPU days to show unsatisfiability.
Furthermore, we created DRAT certificates. 
The certificate for $n=10$ is about 78GB and 
the verification with DRAT-trim took about 6 CPU days.
The resources used for $n \le 9$ are negligible.

\paragraph{Plane Hamiltonian subdrawing of $2n-3$ edges (Conjecture~\ref{conjecture:rafla_2n_plus_3})}

To verify that Conjecture~\ref{conjecture:rafla_2n_plus_3} holds
for $n \le 8$, 
we ran
\begin{verbatim}
    for n in {3..10}; do 
        python rotsys_find.py -HC+ $n 
    done
\end{verbatim}
Here the parameter \verb| -HC+| asserts that there is no plane Hamiltonian subdrawing on $2n-3$ edges. 
In the so-created instance, solutions correspond to rotation systems on $[n]$ which do not contain any plane Hamiltonian subdrawing of $2n-3$ edges.
Using the solver CaDiCaL,
it took about 3 CPU days to show unsatisfiability.
Moreover, the verification with DRAT-trim took about 3 CPU days.

\paragraph{Extension of plane Hamiltonian cycles (Conjecture~\ref{conjecture:extend_HC_to_2n_3})}

To verify that Conjecture~\ref{conjecture:extend_HC_to_2n_3} holds for $n\le 10$,
run
\begin{verbatim}
    for n in {3..10}; do 
        python rotsys_find.py -HC++ -nat -c $n 
    done
\end{verbatim}
Here the parameter \verb| -HC++| asserts that there is a plane Hamiltonian cycle $C$
that cannot be extended to a plane Hamiltonian subdrawing on $2n-3$ edges. 
Note that, since we assume that $C$ traverses $1,2,\ldots,n$ in this order,
we cannot assume natural labeling. 
Therefore, we have to use the \verb|-nat| parameter 
which allows the framework to also consider rotation systems 
which are not natural.  
Using the solver CaDiCaL,
it took about 2 CPU days to show unsatisfiability.
Moreover, the verification with DRAT-trim took about 2 CPU days.

\paragraph{Hamiltonian cycles avoiding a matching (Conjecture~\ref{conjecture:hoffmanntoth_convex})}

To verify that Conjecture~\ref{conjecture:hoffmanntoth_convex} holds for $n \le 11$
run
\begin{verbatim}
    for n in {4..11}; do
        for (( k = 1; 2*k <= n; k++ )); do
            python rotsys_find.py $n -c -nat -HT+ $k
        done
    done
\end{verbatim}
With the parameter \verb|-HT+ k| we
assume towards a contradiction that there exists a convex drawing of $K_n$ 
which has a plane matching $M=\{1,2\},\ldots,\{2k-1,2k\}$
such that for every plane Hamiltonian cycle $C$ crosses edges of $M$,
that is, $C \cup M$ (not necessarily disjoint union) is non-plane.
Note that for $k \ge 2$ we cannot assume without loss of generality that the rotation system is natural, that is, $1$ sees $2,3,\ldots,n$ in this order.
To optimize the encoding and to speed up the computations, 
observe that
we can assume without loss of generality that 
\begin{itemize}
	\item
	$1$ sees $2,u,v$ for every edge $\{u,v\} \in M$ with $u<v$,
	\item $1$ sees $2,u,v'$ for every two edges  $\{u,v\},\{u',u'\} \in M$ with $u<v$, $u'<v'$ and $u<u'$,
	\item $1$ sees $2,x,y$ for every $x,y \in [n] \setminus V(M)$ with $x<y$. 
\end{itemize}

Using the solver CaDiCaL,
the computations for $n \le 11$ took about 12 CPU days. Moreover, the certification with DRAT-trim took about 12 CPU days.

\subsection{Uncrossed edges}
\label{uncrossed}

To verify that every rotation system on $n=7$ contains an uncrossed edge, use the following command: 
\begin{verbatim}
    for n in {4,5,6,7}; do
        python rotsys_find.py -aec $n 
    done
\end{verbatim}

To find a rotation system on $n=8$ which does not contain an uncrossed edge, use the following command:
\begin{verbatim}
    python rotsys_find.py -aec 8 
\end{verbatim}

To restrict to convex (resp.\ $h$-convex) drawings one needs to add the additional parameter \verb| -c | (resp.\ \verb|-hc|).
We provide $h$-convex rotation systems that verify Conjecture~\ref{conjecture:hconvex_aec} for $n \le 21$ 
as supplemental files in the folder 
\verb|examples/all_edges_crossing/|.

\subsection{Empty triangles}
\label{emptytriangles}

To search for a drawing with at most $k$ empty triangles,
we used \linebreak \verb|pysat.card.CardEnc.atmost| 
from PySAT \cite{pysat}.
This method creates constraints for a CNF
that assert that among a given set of variables 
(here we use the $\triangle$ variables)
at most $k$ are $true$.

To verify that $2n-4$ is a tight lower bound on the number of triangles for $n \le 9$, use:
\begin{verbatim}
    for n in {4..9}; do 
        python rotsys_find.py -etupp $((2*$n-5)) $n
    done
\end{verbatim}
Here the parameter \verb|-etupp k| asserts that the number of empty triangles is at most~$k$. Since all instances are unsatisfiable for $k=2n-5$,
it follows $\triangle \ge 2n-4$.

We provide the examples as supplemental files in 
\verb|examples/empty_triangles/|.

\end{document}